\documentclass[journal]{IEEEtran}

\usepackage{cite}
\usepackage{graphicx}
\usepackage[cmex10]{amsmath}
\usepackage{amssymb,amsfonts}
\usepackage{algorithmic}
\usepackage{algorithm}
\usepackage{array}
\usepackage{url}
\usepackage{color}
\usepackage{setspace}

\newtheorem{lemma}{Lemma}

\newtheorem{proposition}{Proposition}
\newtheorem{conjecture}{Conjecture}

\newtheorem{remark}{Remark}
\newenvironment{proof}{{\indent\emph{Proof:}}}{$\hfill\Box$\vspace{.1in}}

\begin{document}
\title{Optimal Relaying in \\Energy Harvesting Wireless Networks with Wireless-Powered Relays}

\author{Masoumeh~Moradian,~\IEEEmembership{Student Member,~IEEE},
        Farid~Ashtiani,~\IEEEmembership{Member,~IEEE}, and Ying Jun (Angela) Zhang,~\IEEEmembership{Senior Member,~IEEE}       
        
\thanks{M. Moradian and F. Ashtinai are with the Advanced Communications Research Institute (ACRI), Department of Electrical Engineering, Sharif University of Technology, Tehran, Iran (email: mmoradian@ee.sharif.edu, ashtianimt@sharif.edu).

This work has been initiated when the first author was a visiting student at the Department of Information Engineering, the Chinese University of Hong Kong. 

Y. J. Zhang is with the Department of Information Engineering, the Chinese University of Hong Kong. She is also with the Institute of Network Coding (Shenzhen), Shenzhen Research Institute, the Chinese University of Hong Kong (e-mail: yjzhang@ie.cuhk.edu.hk).   }
}




\maketitle
\thispagestyle{empty}
\begin{abstract}
In this paper, we consider a wireless cooperative network with an energy harvesting relay which is powered by the energy harvested from ambient RF waves, such as that of a data packet. At any given time, the relay operates either in the energy harvesting (EH) mode or the data decoding (DD) mode, but not both. Separate energy and data buffers are kept at the relay to store the harvested energy and decoded data packets, respectively. In this paper, we optimize a time switching policy that switches between the EH mode and DD mode to maximize the system throughput or minimize the average transmission delay. Both static and dynamic time switching policies are derived. In particular, static policies are the ones where EH or DD mode is selected with a pre-determined probability. In contrast, in a dynamic policy, the mode is selected dynamically according to the states of data and energy buffers. We prove that the throughput-optimal static and dynamic policies keep the relay data buffer at the boundary of stability. More specifically, we show that the throughput-optimal dynamic policy has a threshold-based structure. Moreover, we prove that the delay-optimal dynamic policy also has a threshold-based structure and keeps at most one packet at the relay. We notice that the delay-optimal and throughput-optimal dynamic policies coincide in most cases. However, it is not true for optimal static policies. Finally, through extensive numerical results, we show the efficiency of optimal dynamic policies compared with the static ones in different conditions. 

\end{abstract}


\section{Introduction} \label{sec:intro}
\subsection{Background and contributions}
\IEEEPARstart{E}{nergy} harvesting (EH), a technology to collect energy from the surrounding environment, has received considerable attention as a sustainable solution to prolong the lifetime of wireless networks under energy constraints. EH technology is especially appealing in networks with low-cost battery-powered devices, e.g., wireless sensor networks \cite{dl:sude}. Unlike battery-powered networks, energy-harvesting wireless networks potentially have an unlimited lifetime, thanks to the unlimited energy supply in the environment. Thus, with relatively low cost, a set of EH relay nodes can be overlaid on an existing non-EH network to enhance the performance and quality of service (QoS). 

In addition to conventional harvestable energy sources such as solar, wind, and thermal energy, ambient RF waves have emerged as a new source of energy for wireless EH (WEH) nodes. Since radio waves carry both information and energy, simultaneous wireless information and power transfer (SWIPT) has been actively studied in \cite{dl:varshney,dl:grover}, which assume that a WEH node can decode information and harvest energy simultaneously from the same radio signal. However, due to practical hardware and circuit limitation, a radio signal that has been used for energy harvesting may not be reused for data decoding \cite{dl:huang}. Correspondingly, time-switching (TS) protocols, which switch between EH and data decoding (DD) modes, have been studied in single-antenna point-to-point channels \cite{dl:liu1}, in two-user MIMO interference channels \cite{dl:park}, and in amplify-and-forward \cite{dl:nasir1,dl:nasir2} and decode-and-forward \cite{dl:gu,dl:wang} relaying networks.  

In this paper, we focus on a cooperative wireless network with a WEH relay that uses a TS protocol to switch between the EH and DD modes. In the DD mode, the relay decodes the source packets that are not successfully transmitted to the destination and stores them in its data buffer. In the EH mode, the relay harvests energy from the ambient RF waves, e.g., incoming source packets or co-channel interference \cite{dl:liu1} to replenish its energy buffer. For such a system, we aim to find the optimal TS policies that maximize the system throughput or minimize the average end-to-end transmission delay. A key challenge of this work lies in the inherent correlation between the arrival and departure processes of the data and energy buffers at the WEH node. 
The correlation between the departure processes is apparent in that a data packet departure must be accompanied by a simultaneous departure of energy packets, due to the energy consumption for data transmission. In addition, the correlation in the arrival processes is due to the TS policy, where packets arrive at either the data buffer or energy buffer depending on whether a node is in the DD mode or EH mode. Note that an imbalance in the backlog of data and energy buffers would result in a degradation of system performance, because cooperative relaying is possible only when both buffers are backlogged. The intrinsic correlation between the data and energy buffers complicates the analysis of the system. To tackle this challenge, we propose a quasi-birth-death (QBD) process to analyze the throughput and delay as a function of the TS policy. Optimal static and dynamic TS policies are then derived for throughput maximization and delay minimization, respectively. 

The main contributions in this paper are summarized as follow:
\begin{itemize}
\item We derive optimal static TS policies, in which modes are selected based on a pre-determined probability regardless of the states of the buffers. We prove that the throughput-optimal static policy is the one that keeps the data buffer at the WEH relay at the boundary of stability. Furthermore, the delay-optimal static policy is obtained by analyzing an underlying QBD process. We also derive the necessary and sufficient condition under which the non-cooperation policy is delay-optimal.
 
\item We derive dynamic TS policies, where modes are selected dynamically based on the states of the data and energy buffers. We prove that both the throughput-optimal and delay-optimal policies are threshold-based in terms of the status of the energy buffer. Moreover, the delay-optimal dynamic policy is the one that keeps at most one packet in the data buffer. Interestingly, we observe that the throughput-optimal and delay-optimal dynamic TS policies are the same in most cases. However, this is not the case for static TS policies.
 
\item By simulations, we validate our analyses and compare the performance of the optimal static and dynamic policies in different conditions. 
\end{itemize}

For the convenience of readers, the throughput-optimal and delay-optimal static and dynamic policies are summarized in Table \ref{tab:table2} in Section \ref{Delay-optimal dynamic policy}. 

\subsection{Related works}
To the best of our knowledge, our work is the first one that explores the TS policy in a WEH relay with both data and energy buffers. In \cite{dl:nasir1,dl:nasir2,dl:gu,dl:wang,dl:nasir3} the maximum achievable throughput of the source is derived in WEH cooperative networks that apply TS policy at the relay node. However, none of them include a data buffer at the relay node. 
In fact, the information which is decoded by the relay in a slot, is transmitted afterwards in the same slot using the energy harvested from source transmissions. However, the energy harvested in a slot may not be enough for transmission of the packet. By including a data buffer at the relay, the packet can be stored until enough energy is harvested.
In \cite{dl:nasir1}, the required energy for transmission in a slot is harvested in a variable portion of the same slot, i.e., continuous EH. The same authors investigated the discrete EH in \cite{dl:nasir2} where the required energy for transmission in a slot is harvested in previous slots. They also compared these two schemes in \cite{dl:nasir3}. In \cite{dl:gu}, the relay harvests energy from interference signals as well as source transmissions in a continuous EH manner. The authors in \cite{dl:wang} considered discrete EH in a network with multiple relays where the relays cooperate to transmit the information of the source.  

The delay performance of cooperative networks with EH relays has been considered in \cite{dl:mine,dl:rajib,dl:tang}. In \cite{dl:mine,dl:rajib}, the average transmission delay in a wireless network with cooperative EH relay(s) is minimized. The relay, however, harvests energy from the environmental resources instead of RF waves, and thus the arrival processes of data and energy buffers are not correlated. 
The authors in \cite{dl:tang} minimize the transmission completion time of a fixed amount of data backlogged at the source in a cooperative network with multiple WEH relays. Since the data backlogged at the source is limited, the optimization includes the order of transmissions and their corresponding durations.
However in our scenario, the source is always backlogged, i.e., there are always unsuccessfully transmitted packets that can be stored and transmitted by the relay. Thus, unlike \cite{dl:tang}, the data buffer at node R may become unstable due to applying an inappropriate TS policy. In this respect, we optimize the long-term performance metrics, e.g., the average transmission time of the packets while preserving the stability of the data buffer at the relay node.     

The rest of the paper is organized as in the following. In Section II, we introduce the system model, the underlying QBD process, and the desired performance metrics. Section III is dedicated to the derivation of optimal static policies. Section IV discusses how we find the optimal dynamic policies. In Section V, the numerical results are presented. Finally, Section VI concludes the paper. Also, the notations are introduced in Table \ref{tab:table1}.
\begin{table}[!t]
\begin{center}
\caption{Parameters}
\scalebox{0.75}{
    \begin{tabular}{ | c | c |}
    \hline
    Parameter                &  Explanation                               \\ \hline
    $Q_d$                    &  Data buffer at node R                   \\ \hline
    $Q_e$                   &  Energy buffer at node R       \\ \hline
    $p^{det}_S$ ($p^{det}_R$)&  Detection probability of S-D (R-D) link \\ \hline
    $\gamma_S$ ($\gamma_R$)  &  Energy arrival distribution at node S (node R)  \\ \hline 
    $N$ & Size of the energy buffer, $\text{Q}_\text{e}$ \\ \hline
    $K$ & The number of energy units used for a packet transmission \\ \hline
    $b_{\text{max}}$ & The maximum number of energy units harvested in a slot \\ \hline
    $\Gamma$ & The random variable denoting the number of energy units \\ 
    & harvested in a slot in EH mode \\ \hline 
    $\gamma_i$ & The probability of harvesting $i$ energy units in a slot \\ \hline
    $\tau$                   &  Average transmission delay of source packets      \\ \hline
    $q_d$& The number of data packets backlogged at node $R$ \\ \hline
    $q_e$& The number of energy units backlogged at node $R$ \\ \hline
    $\overline{D_R}$         &  Average system delay at node R  \\ \hline
    $\lambda_{id}$ & Data arrival rate at node R \\ \hline
    $\lambda_{od}$ & Data departure rate at node R \\ \hline
    $\lambda_{ie}$ & Energy arrival rate at node R after blocking \\ \hline
    $\lambda_{oe}$ & Energy departure rate at node R \\ \hline
    $\alpha_s$ & The probability of switching to the DD mode at state $s$ \\ \hline
    $p^a_R$ & The probability that node $R$ is active \\ \hline
    $p_b$ & The blocking probability of energy buffer, $\text{Q}_\text{e}$ \\  \hline
    $\alpha^{T}$ & Throughput-optimal data decoding probability in static policy \\ \hline 
    $\alpha^{D}$ & Delay-optimal data decoding probability in static policy \\ \hline
    $\alpha^{T}_{s}$ & Throughput-optimal dynamic policy \\ \hline
    $\alpha^{D}_{s}$ & Delay-optimal dynamic policy \\ \hline
    \end{tabular}}
    \label{tab:table1}
\end{center} 
\end{table} 

\section{System Model}
\label{sec:Network scenario}
\subsection{Network scenario}
\label{subsec:system setup}
We consider a three-node cooperative wireless network comprised of a source (S), a relay (R), and a destination (D) node, as shown in Fig. \ref{fig:scenario}(a). All nodes are within the transmission range of each other. Suppose that nodes S and D are battery-powered and node S is infinitely backlogged, i.e., it has always a packet to transmit. On the other hand, node R is an energy harvesting (EH) node, which harvests its energy from ambient RF waves, e.g., source transmissions, an RF energy source or interference signals and is able to store the harvested energy in a rechargeable battery. Suppose that node R is equipped with one antenna and two different circuits, one for EH and the other for DD. Thus, node R can only transmit data packets, receive data packets, or harvest energy at a particular time. Node R may relay the data packets of the source, but it does not have its own traffic. The data and energy buffers at node R are denoted by $\text{Q}_\text{d}$ and $\text{Q}_\text{e}$, respectively. While the data buffer is assumed to have infinite size, we assume that the energy buffer size is finite due to the finite battery capacity in practice. The assumption of infinite data buffer is made to analyze the stability of the data buffer. However, it can be considered sufficiently large in practice.  

\begin{figure}[!t]
\begin{center}
\includegraphics[height=130pt,width=\columnwidth]{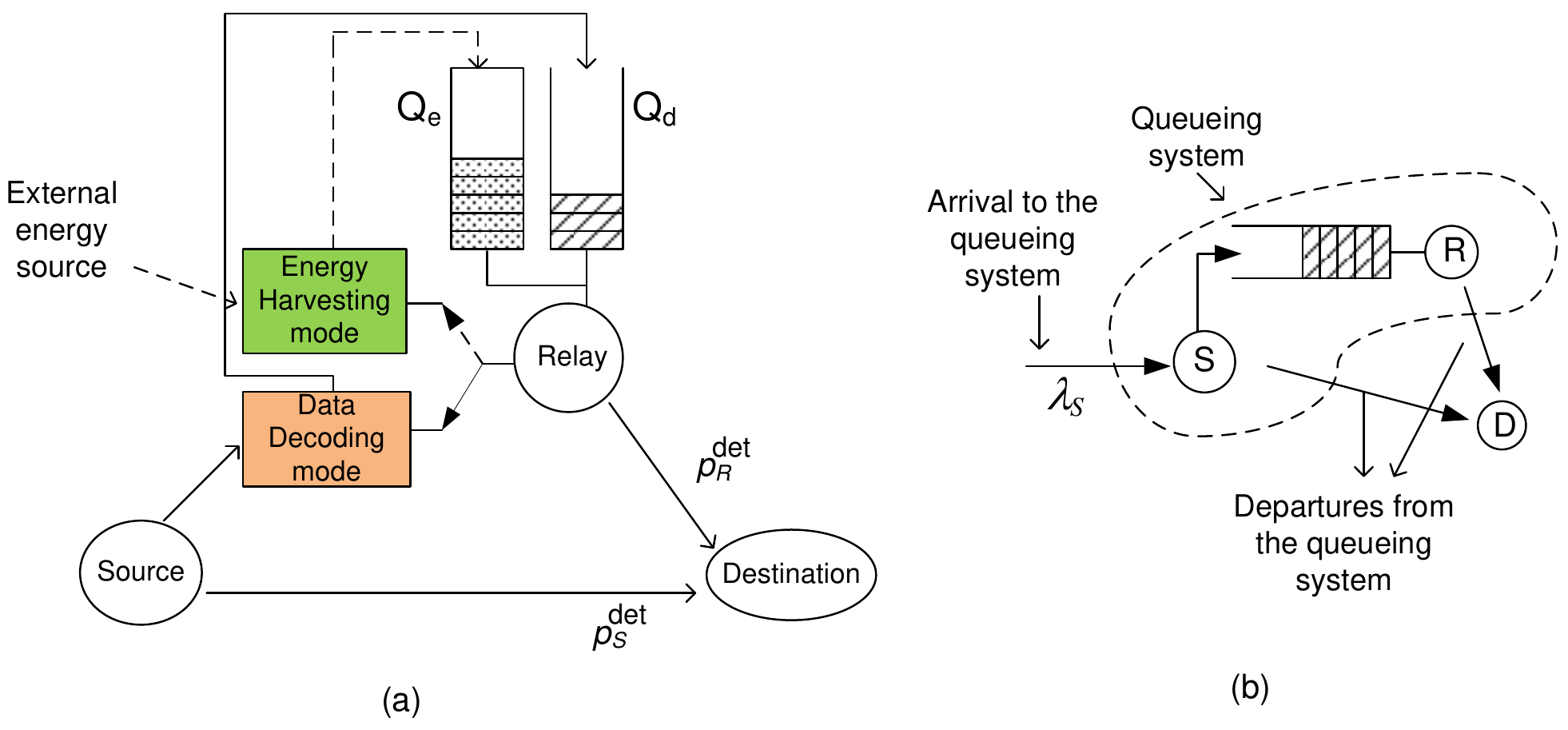}
\caption{(a) Network scenario, (b) Equivalent queueing system.}
\label{fig:scenario}
\end{center}
\end{figure}

We assume that the channel between nodes S and R (i.e., S-R channel) is ideal. In practice, this means that node R is sufficiently close to node S, so that it is able to harvest noticeable energy from the source transmissions. However, the channels between S and D (i.e., S-D channel) and between R and D (i.e., R-D channel) are non-ideal with detection probabilities $p^{det}_{S}$ and $p^{det}_{R}$, respectively. Moreover, time is considered to be slotted, and each time slot is further divided into two equal subslots. Without loss of generality, in our scenario, a time slot is considered as the time unit. The transmission strategies of nodes S, R and D are described as follows.

\emph{Node S}: Node S transmits a data packet in the first subslot and remains silent in the second subslot of each time slot. If it receives an ACK from either node R or node D, then the current packet is removed from its buffer, and it will transmit a new packet in the next time slot. Otherwise, the current packet will be retransmitted in the next time slot. 

\emph{Node R}: At the beginning of each slot, node R decides whether to operate in the DD mode or EH mode. When operating in the DD mode, node R decodes the packet of node S in the first subslot. At the end of the first subslot, node R discards the received packet if an ACK is received from node D. Otherwise, it stores the packet in its data buffer, $\text{Q}_\text{d}$, and sends an ACK to node S. On the other hand, if node R is in the EH mode, it harvests energy from the transmission of node S in the first subslot. In either mode, node R transmits in the second subslot when it has at least one data packet in its data buffer and enough energy units backlogged in its energy buffer. In this case, we say that node R is in the active state.

Suppose that node R consumes $K$ units of energy to transmit a data packet. The data packets are transmitted on a first-come first-serve (FCFS) basis. If a packet is not transmitted successfully in a slot, it will remain at the head-of-line of the data buffer and be retransmitted in a following time slot whenever there exists at least $K$ energy units in the energy buffer.

\begin{figure}[!t]
\begin{center}
\includegraphics[height=130pt,width=\columnwidth]{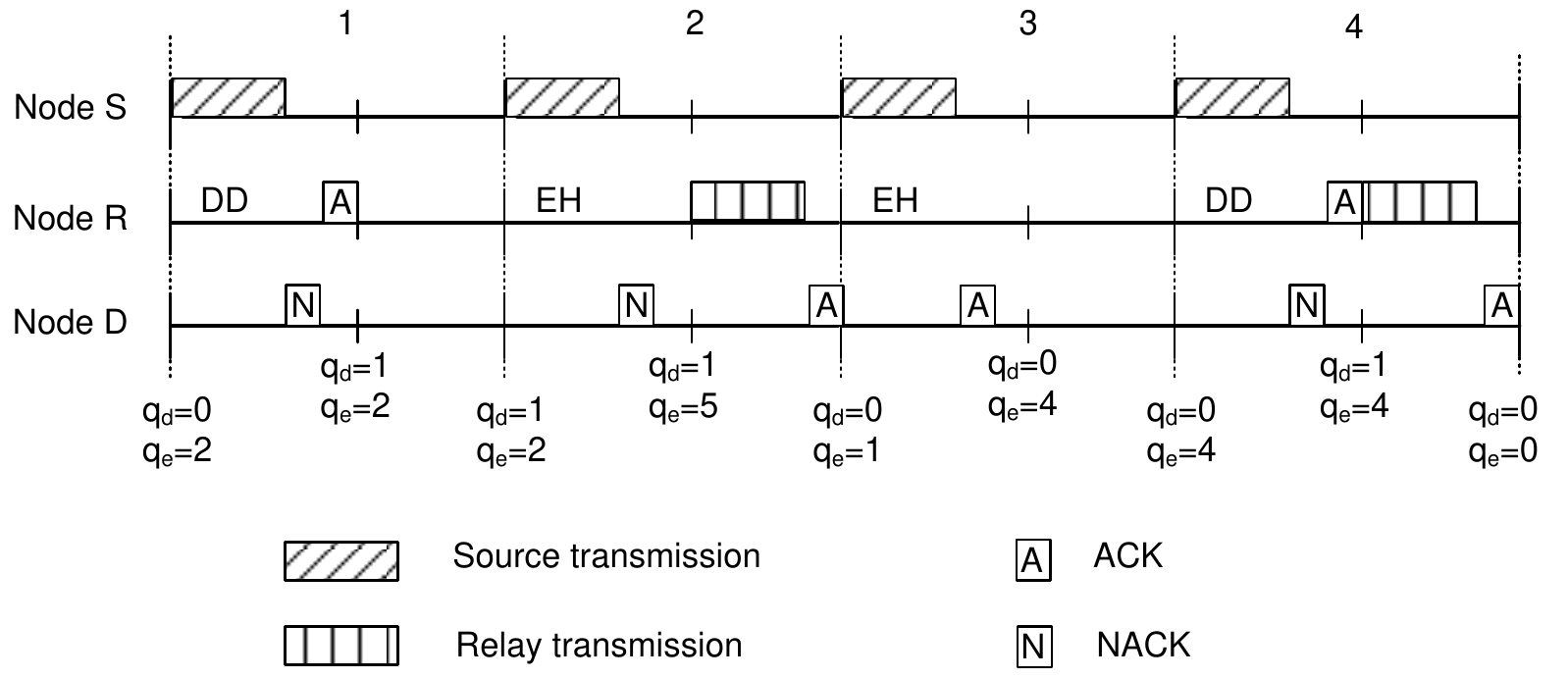}
\caption{Transmission protocol ( $b_{\max}=3$, $K=4$)}
\label{fig:protocol}
\end{center}
\end{figure}

\emph{Node D}: Upon receiving a data packet successfully, it broadcasts an ACK message. 

A realization of the transmission strategies of the nodes is illustrated in Fig. \ref{fig:protocol}. Suppose that there are initially $2$ energy units stored in $\text{Q}_\text{e}$. As shown in Fig. \ref{fig:protocol}, in the first slot, node R is in the DD mode. Upon hearing NACK in the first subslot, it stores the source packet in $\text{Q}_\text{d}$ and transmits an ACK. However, due to lack of energy, it cannot transmit the packet in the second subslot. In the second slot, node R switches to the EH mode, harvests enough energy in the first subslot and transmits the backlogged source packet successfully in the second subslot. In the third slot, node R remains in the EH mode to harvest more energy. Finally, it switches to the DD mode in the forth slot, receives a source packet and transmits it successfully in the same slot.

Let $\Gamma$ be the random variable denoting the number of energy units harvested in a slot when node R is in the EH mode. Suppose that $\Gamma$ is independently and identically distributed over different slots. Define $\gamma_m=Pr\{\Gamma=m\}$ for $ 0 \leq m \leq b_{\text{max}}$, where $b_{\text{max}}$ is the maximum number of energy units harvested in a slot. 
The stochastic nature of the harvested energy is due to the randomness in wireless channels and EH circuitry. Moreover, the energy harvested from RF waves at each slot is generally low compared to the amount of energy required for transmission of a packet \cite{dl:lu}. This is because the RF waves are severely attenuated due to path loss and fading. Thus, it is reasonable to assume that $b_{\text{max}}\leq K$. In addition, we assume that the size of the energy buffer, i.e., $N$, is large enough to store the energy needed for at least two transmissions, i.e., $N\geq 2K$.

\begin{figure}[!t]
\begin{center}
\includegraphics[height=150pt]{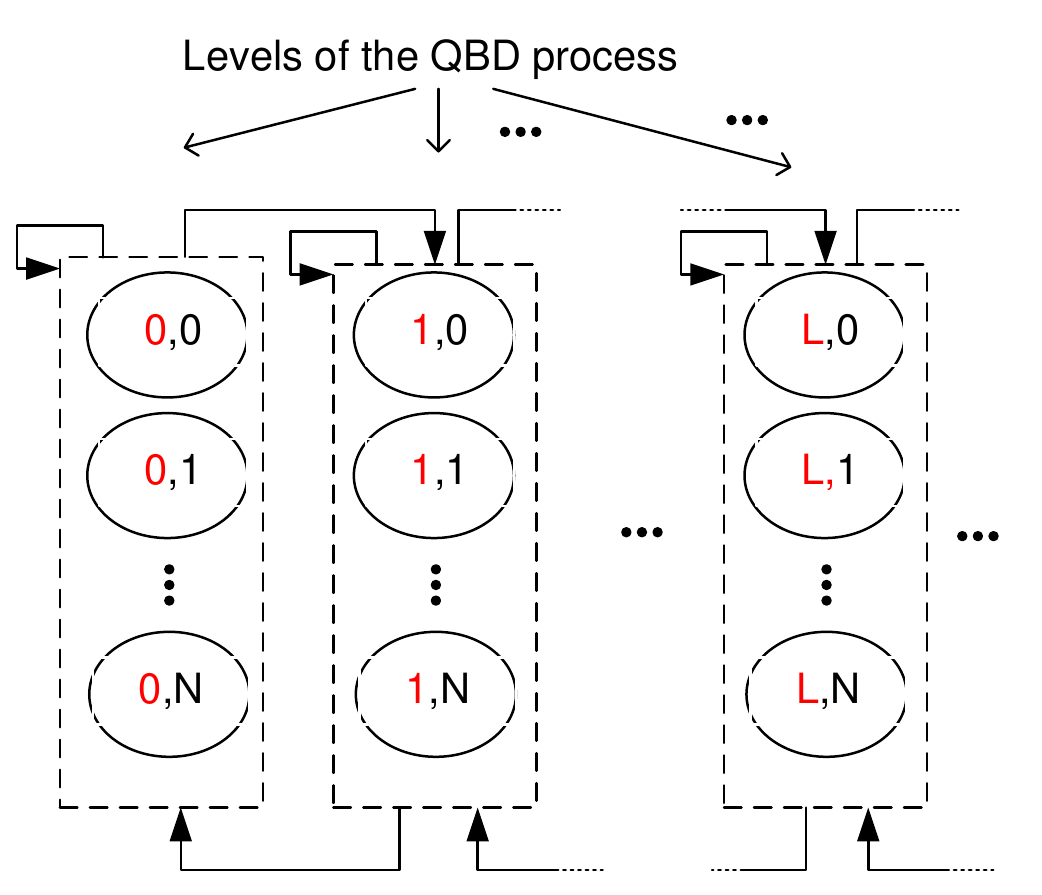}
\caption{The Underlying QBD model.}
\label{fig:qbd}
\end{center}
\end{figure}

\subsection{Mode selection and the underlying QBD process}
\label{Mode selection and underlying QBD process}
Define the system state at the beginning of each time slot by an ordered-pair $(q_d,q_e)$, where the first and the second components denote the number of data packets and energy units backlogged at node R, respectively. At each slot, based on the current state, node R switches to either DD or EH mode. The system state evolves according to a Markovian process and can be modeled by a QBD process shown in Fig. \ref{fig:qbd}. According to the convention, the first entry of system state denotes the level of the QBD and the second entry denotes its phase \cite{dl:lat}. There are an infinite number of levels and a finite number of phases in the QBD, corresponding to our assumption on infinite data buffer size and finite energy buffer size. Moreover, transitions occur between the states only within the same or adjacent levels. This corresponds to our assumption that at most one data packet arrives at or departs from $\text{Q}_\text{d}$ at each slot.

The transition probability from a state in the first level of the QBD, i.e., $(0,q_e)$, to other states, denoted by $P_{0i\rightarrow lj}$, is derived as (please refer to Appendix \ref{ap:qbd} for more details)
\begin{equation}
P_{0i\rightarrow lj}=\begin{cases}
\alpha_s(1-p^{det}_{S})p^{det}_{R} \mathbf{M}_{ij}+\alpha_s p^{det}_{S} \mathbf{I}_{ij} \\
\quad +(1-\alpha_s)\mathbf{T}_{ij} ~\qquad \qquad \qquad;\qquad \quad  l=0~, \\
\alpha_s(1-p^{det}_{S})(1-p^{det}_{R}) \mathbf{M}_{ij}\qquad; \quad\qquad ~  l=1~, \\
\end{cases}
\label{eqtp1}
\end{equation}
where $\alpha_s$ denotes the probability of switching to the DD mode at state $s$. $\mathbf{M}$ is derived in \eqref{eqM2} and represents the transition matrix of the energy state in the second subslot, given that the data buffer at node R is backlogged at the beginning of the second subslot. Likewise, the energy state in the first subslot changes according to the transition matrix $\mathbf{T}$ in \eqref{eqT1}, given that node R selects EH mode. Furthermore, the transition probability from state $s=(l,i)$ to state $(l',j)$ ($l>0$), represented by $P_{li\rightarrow l'j}$, is written as 
\begin{equation}
P_{li\rightarrow l'j}=\begin{cases}
\alpha_s((1-p^{det}_{S})p^{det}_{R}+p^{det}_{S}(1-p^{det}_{R}))\mathbf{M}_{ij}\\
+(1-\alpha_s)(1-p^{det}_{R})\mathbf{B}_{ij}\qquad\qquad~ ;~~ l'=l~, \\
\alpha_s(1-p^{det}_{S})(1-p^{det}_{R})\mathbf{M}_{ij} \qquad~~~ ;  ~  ~l'=l+1, \\
\alpha_s p^{det}_{S}p^{det}_{R}\mathbf{M}_{ij}+(1-\alpha_s)p^{det}_{R}\mathbf{B}_{ij} ~;  ~ l'=l-1. 
\end{cases}
\label{eqtp2}
\end{equation}
where $\mathbf{B}$ is the transition matrix of the energy buffer state in a slot, i.e., in the first and second subslots, given that the data buffer at node R is backlogged at the beginning of the first subslot and node R selects EH mode. In particular, $\mathbf{B}=\mathbf{T} \times \mathbf{M}$. It is worth noting that including the effect of non-ideality of S-R channel is straightforward by adding a detection probability in the related transition probabilities.

\subsection{Throughput and average transmission delay}
In this subsection, we define two performance metrics, namely, source throughput and average transmission delay. In particular, source throughput, denoted by $\lambda_S$, is defined as the rate at which the packets are successfully detected at node D. By the assumption that node S is infinitely backlogged, $\lambda_S$ can be derived as
\begin{equation}
\lambda_S = p^{det}_S+p^a_R p^{det}_R~,
\label{eq:thr}
\end{equation}
where $p^a_R$ is the probability that node R is active, i.e., it has at least one data packet and $K$ energy units backlogged at the beginning of the second subslot. 
Likewise, the average transmission delay is defined as the average time from the moment when a packet becomes HoL at node S until it is successfully detected at node D. Regarding nodes S and R as a queueing system shown in Fig. 1(b), we can derive the average system delay, $\tau$, according to the Little's Law [23] as 
\begin{equation}
\tau = \frac{\overline{q_{d}}+1}{\lambda_{S}}~,
\label{eq:delay}
\end{equation}
where $\overline{q_{d}}$ is the average number of data packets at node R. Here, the average number of packets in the whole system is equal to $\overline{q_{d}}+1$, because node S is infinitely backlogged and thus always has one packet in its server. Moreover, the arrival rate to the queueing system is equal to the throughput $\lambda_S$ when the system is stable. Our numerical results in Fig. \ref{fig:confirm} confirm the validity of \eqref{eq:thr} and \eqref{eq:delay}.

It is worth to mention that the policies achieving the maximum throughput and minimum delay, respectively, are not necessarily the same. To see this, note that the source throughput is maximized when $p_R^a$ is maximized. That is, the probability that the data buffer $\text{Q}_\text{d}$ is backlogged and the energy buffer $\text{Q}_\text{e}$ has at least $K$ energy units, is maximized. However, this does not necessarily minimize $\overline{q_d}$, the average number of data packets backlogged at node R. The throughput-optimal and delay-optimal TS policies are derived in subsequent sections. 

\section{Optimal Static Policies}
Under a static TS policy, node R selects an operation mode with a fixed probability irrespective of the states of the data and energy buffers. That is, $\alpha_s=\alpha~ \forall s$, where $\alpha_s$ was defined in Section \ref{Mode selection and underlying QBD process}. In this section, we derive the optimal $\alpha$ that maximizes the throughput and minimizes the average transmission delay, respectively. 

\subsection{Throughput-optimal static policy}
\label{Throughput-optimal static policy}
By the following lemmas, we prove that the throughput-optimal static policy is the policy that keeps $\text{Q}_\text{d}$, the data buffer at node R, at the boundary of stability\footnote{The boundary of stability of a queue is defined as the point in which if the arrival rate is increased by any $ \epsilon>0$, the queue becomes unstable i.e., its arrival rate exceeds its departure rate.}. 


\begin{lemma}
$p^a_R$ is an increasing function of $\alpha$ when $\text{Q}_\text{d}$ is stable.
\label{lem:stability}
\end{lemma}
\begin{proof}
Define $\lambda_{id}$ and $\lambda_{od}$ to be the arrival and departure rates at $\text{Q}_\text{d}$, respectively. Then
\begin{equation}
\lambda_{id}=\alpha (1-p^{det}_S),
\label{eqlamid}
\end{equation}
\begin{equation}
\lambda_{od}=p^a_R p^{det}_R.
\label{eqlamod}
\end{equation}
 When $\text{Q}_\text{d}$ is stable, its arrival rate is equal to its departure rate, i.e., $\lambda_{id}=\lambda_{od}$. Thus, from \eqref{eqlamid} and \eqref{eqlamod}, $p^a_R$ is derived as
\begin{equation}
p^a_R=\frac{\alpha (1-p^{det}_S)}{p^{det}_R},
\label{eqparalpha}
\end{equation}
which is obviously an increasing function of  $\alpha$.
\end{proof}
\begin{lemma}
The blocking probability of the energy buffer $\text{Q}_\text{e}$, denoted by $p_b$, is zero if and only if $\text{Q}_\text{d}$ is at the boundary of stability or unstable.
\label{lem:pb}
\end{lemma}
\begin{proof}
We first prove the ``if" part. When $\text{Q}_\text{d}$ is at the boundary of stability or unstable, there is always a packet in $\text{Q}_\text{d}$. In fact, since in static policy $\alpha_s=\alpha~ \forall s$, the arrival process at $\text{Q}_\text{d}$ is a Bernoulli process, thus similar to a Geo/G/1 queue, $\text{Q}_\text{d}$ will be always backlogged at the boundary of stability. Thus, at each time slot, if $q_e \geq K$, then $K$ energy units are consumed to transmit a packet. Since $b_{\text{max}}\leq K$ and $N \geq 2K$ (see Section \ref{subsec:system setup}), the probability that the energy buffer is full is zero, resulting in a zero blocking probability at $\text{Q}_\text{e}$.

Now, we prove the ``only if" part by contradiction. Let $\lambda_{ie}$ denote the energy arrival rate after blocking at the energy buffer. Likewise, define $\lambda_{oe}$ to be the departure rate from $\text{Q}_\text{e}$. Then, we have
\begin{equation}
\lambda_{oe}=\lambda_{ie}=(1-\alpha)E(\Gamma)(1-p_b) 
\label{eq:lamoe}
\end{equation}
where $E(\Gamma)$ is the average of $\Gamma$. Now, suppose $p_b=0$ when $\text{Q}_\text{d}$ is stable. Then, \eqref{eq:lamoe} becomes $\lambda_{oe}=\lambda_{ie}=(1-\alpha)E(\Gamma)$, implying that $\lambda_{oe}$ and $\lambda_{ie}$ are decreasing functions of $\alpha$. Note that a transmission attempt of a data packet at node R consumes $K$ energy units at the same time. Thus, a decrease in $\lambda_{oe}$ is equivalent to a decrease in the data departure rate, i.e., $\lambda_{od}$. On the other hand, due to stability, the arrival and departure rates at $\text{Q}_\text{d}$ are equal, i.e., $\lambda_{od}=\lambda_{id}=\alpha (1-p^{det}_S)$, implying that $\lambda_{od}$ increases with $\alpha$. This is a contradiction, except when the equation $\lambda_{od}=\lambda_{id}$ does not hold by increasing $\alpha$, i.e., $\text{Q}_\text{d}$ is at the boundary of stability, or when the data arrival and departure rates are not equal even before increasing $\alpha$, i.e., $\text{Q}_\text{d}$ is unstable.

\end{proof}
\begin{lemma}
$p^a_R$ is a decreasing function of $\alpha$ when $\text{Q}_\text{d}$ is unstable.
\label{lem:unstab}
\end{lemma}
\begin{proof}
A packet departure from $\text{Q}_\text{d}$ consumes $\frac{K}{p^{det}_R}$ energy units on average, since each packet is transmitted $\frac{1}{p^{det}_R}$ times on average. Thus, the energy departure rate, $\lambda_{oe}$, can be written as a function of data departure rate, $\lambda_{od}$, as follows:
\begin{equation}
\lambda_{oe}= \lambda_{od} \frac{K}{p^{det}_R}.
\label{eq:eqlamoe1}
\end{equation}  
From \eqref{eqlamod}, \eqref{eq:lamoe} and \eqref{eq:eqlamoe1}, we have
\begin{equation}
\alpha= 1-\frac{p^a_R K}{E(\Gamma)(1-p_b)}.
\label{eq:eqlamoe2}
\end{equation}  
According to Lemma \ref{lem:pb}, $p_b=0$ when $\text{Q}_\text{d}$ is unstable. Thus, it can be seen from \eqref{eq:eqlamoe2} that $p^a_R$ is decreasing with $\alpha$ when $\text{Q}_\text{d}$ is unstable. 
   
\end{proof}

In the following proposition, we derive the optimum $\alpha$ that maximizes the throughput, represented by $\alpha^{T}$, and show that it keeps $\text{Q}_\text{d}$ at the boundary of stability.
\begin{proposition}
The throughput-optimal static policy keeps $\text{Q}_\text{d}$ at the boundary of stability. In addition, $\alpha^{T}$ is given by
\begin{equation}
\alpha^{T} = \frac{E(\Gamma) p^{det}_{R}}{E(\Gamma) p^{det}_{R} +K(1-p^{det}_S)}.
\label{eq:alpha}
\end{equation}
\label{pro:alpha}
\end{proposition}
\begin{proof}
Lemmas \ref{lem:stability} and \ref{lem:unstab} show that $p_R^a$ increases with $\alpha$ when $\text{Q}_\text{d}$ is stable and decreases with $\alpha$ when $\text{Q}_\text{d}$ is unstable. Thus, the maximum of $p_R^a$ which leads to the maximum throughput according to \eqref{eq:thr}, is attained by carefully setting $\alpha$ so that $\text{Q}_\text{d}$ is at the boundary of stability. 
Now, in order to derive $\alpha$ which keeps $\text{Q}_\text{d}$ at the boundary of stability, i.e., $\alpha^{T}$, we use \eqref{eqparalpha} and \eqref{eq:eqlamoe2}. 
From these equations, $\alpha$ is written in terms of $p_b$ as in the following

\begin{equation}
\alpha = \frac{E(\Gamma) p^{det}_{R}(1-p_b)}{E(\Gamma) p^{det}_{R}(1-p_b) +K(1-p^{det}_S)}~. 
\label{eq:alphat}
\end{equation}
The above equation holds when $\text{Q}_\text{d}$ is stable including the boundary of stability. Indeed, according to Lemma \ref{lem:pb}, $p_b=0$ at the boundary of stability. Therefore, \eqref{eq:alpha} is obtained. 

\end{proof}
\subsection{Delay-optimal static policy}
\label{delay-optimal static policy}
It is interesting to note that the relay cooperation may degrade the average transmission delay in some cases. For example, this may occur when $p_R^{det}<p_S^{det}$. In this subsection, we first introduce the necessary and sufficient condition for the non-cooperation policy to be delay-optimal. Here, the non-cooperation policy refers to the one that sets $\alpha=0$. That is, the relay never decodes a packet, and therefore never helps to relay a packet. We then derive the optimal $\alpha$ that minimizes the average transmission delay, denoted by $\alpha^{D}$, when the non-cooperation policy is not optimal. 

Suppose that $\overline{D_R}$, the average system delay at node R (comprised of queueing delay and transmission time\footnote{The transmission time of a packet at node $R$ starts from the moment it becomes the head-of-line packet in the data buffer ($\text{Q}_\text{d}$) and terminates when it is successfully transmitted.}), is an increasing convex function of $\alpha$, or equivalently the arrival rate at $\text{Q}_\text{d}$. This is a valid assumption, due to the fact that in a queue, the queue length and subsequently the total delay increase as the rate of random arrivals at the queue increases (refer to P-K formula in \cite{dl:klein} as an example). Our simulation in Fig. \ref{fig:confirm2} also confirms the validity of this assumption.


\begin{lemma} 
The non-cooperation policy is delay-optimal, i.e., $\alpha^{D}=0$, if and only if 
\begin{equation}
\overline{D_R}|_{\alpha \rightarrow 0}> \frac{1}{p^{det}_S},
\label{eqcond}
\end{equation}
where $\overline{D_R}|_{\alpha \rightarrow 0}$ is the average system delay at node R when $\alpha \rightarrow 0$. 
\label{lem:lem5}
\end{lemma} 
\begin{proof}
According to Little's law, $\overline{D_R}$ is derived as $\overline{D_R}=\frac{\overline{q_d}}{\lambda_{id}}$. Thus, substituting \eqref{eqlamid} in \eqref{eq:delay}, we have
\begin{equation}
\tau =  \frac{\overline{q_d}+1}{\lambda_S}=\frac{\alpha(1-p^{det}_S)\overline{D_R}+1}{\lambda_S}.
\label{eq:lem31}
\end{equation}
Moreover, since $\text{Q}_\text{d}$ is stable (otherwise, $\overline{q_d}$ and thus, $\tau$ is infinite), $\lambda_{id}=\lambda_{od}$. Thus, from \eqref{eq:thr}, \eqref{eqlamid} and \eqref{eqlamod}, $\lambda_S$ can be written as
\begin{equation}
\lambda_S =  \alpha(1-p^{det}_S)+p^{det}_S.
\label{eq:lem32}
\end{equation}
Substituting \eqref{eq:lem32} to \eqref{eq:lem31}, we have
\begin{equation}
\tau=\frac{\alpha(1-p^{det}_S)\overline{D_R}+1}{\alpha(1-p^{det}_S)+p^{det}_S}.
\label{eqtaue}
\end{equation} 

The average transmission delay in the non-cooperation policy is equal to $\frac{1}{p^{det}_S}$ ($\alpha=0$ in \eqref{eqtaue}). Obviously, non-cooperation is delay-optimal  if and only if for any given $\alpha$, the corresponding $\tau$ is greater than the non-cooperation delay, i.e., $\tau>\frac{1}{p^{det}_S}$. Substituting this to \eqref{eqtaue} leads to the following inequality:
\begin{equation}
\overline{D_R}|_{\alpha}> \frac{1}{p^{det}_S}, ~~~ \forall \alpha.
\label{eqcond1}
\end{equation}
Recall that $\overline{D_R}$ is an increasing function of $\alpha$. Thus, if \eqref{eqcond1} holds when $\alpha \rightarrow 0$, then the inequality holds for all other $\alpha$'s. Therefore, \eqref{eqcond} is the necessary and sufficient condition for the non-cooperation policy to be delay optimal.
\end{proof}
 
If the condition given by \eqref{eqcond} does not hold, then the delay-optimal DD probability is positive ($\alpha^{D}>0$). In particular
\begin{equation}
 \begin{aligned}
 \alpha^{D} = & \underset{\alpha \in [0,\alpha^{T}]}{\text{argmin}}
 & & \tau=\frac{\overline {q_d}+1}{\lambda_S}=\frac{\overline {q_d}+1}{\alpha(1-p^{det}_S)+p^{det}_S}~, \\
 \end{aligned}
 \label{eqminstatic}
\end{equation}
where $\alpha^{T}$ and $\lambda_S$ are derived in \eqref{eq:alpha} and \eqref{eq:lem32}, respectively. Note that according to Proposition \ref{pro:alpha}, for $\alpha>\alpha^{T}$, $Q_d$ is unstable and $\tau$ becomes infinite.

To solve \eqref{eqminstatic}, we first derive $\overline q_d$ as a function of $\alpha$. Then, we prove in Lemma \ref{lem:convex} that $\tau$ is a convex function of $\alpha$. Finally, a numerical convex optimization method can be deployed to solve \eqref{eqminstatic}. To derive $\overline q_d$, observe that by setting $\alpha_s=\alpha ~\forall s$, the related QBD process becomes homogeneous, thus it can be solved by the matrix analytic method \cite{dl:lat}. As such, we have 
\begin{equation}
\overline{q_d} = \boldsymbol{\pi}_0 \sum\limits_{l=0}^{\infty} l \mathbf{R}^l  \mathbf{1}=
\boldsymbol{\pi}_0 \mathbf{R} (\mathbf{I}-\mathbf{R})^{-2} \mathbf{1}~,
\label{eqlength}
\end{equation} 
where $\boldsymbol{\pi}_0$ is the stationary distribution vector of the states $(0,q_e)$, where $q_e \in \{ 0,1,...,N \} $ (i.e., the states of the first level of QBD). 
Likewise, $\mathbf{R}$\footnote{$\mathbf{R}_{ij}$ is the expected number of visits to the state $(n+1,j)$, before a return to level $n$ or previous levels given that the process starts from state $(n,i)$.} is a matrix related to the QBD process \cite{dl:lat}.
\begin{lemma}
$\tau(\alpha)$ is convex for the range of $\alpha$ that yields $\overline{D_R}<\frac{1}{p^{det}_S}$, i.e., when cooperation outperforms non-cooperation (see Lemma \ref{lem:lem5} and \eqref{eqcond1}).
\label{lem:convex}
\end{lemma}
\begin{proof}
From \eqref{eqtaue}, $\frac{d^2\tau}{d\alpha^2}$ is computed as a function of $\overline D_R$ and its derivative as
\begin{equation}
\begin{aligned}
\frac{d^2\tau}{d\alpha^2}&=\\&\frac{d^2\overline{D_R}}{d\alpha^2}\frac{\alpha (1-p^{det}_S)}{\alpha (1-p^{det}_S)+p^{det}_S}+2\frac{d\overline{D_R}}{d\alpha}\frac{p^{det}_S(1-p^{det}_S)}{(\alpha (1-p^{det}_S)+p^{det}_S)^2}\\
&+2(1-p^{det}_S)^2\frac{1-p^{det}_S\overline{D_R}}{(\alpha (1-p^{det}_S)+p^{det}_S)^3}.
\end{aligned}
\label{eq:secondder}
\end{equation}
Since $\overline{D_R}$ is assumed to be convex and increasing in terms of $\alpha$, then $\frac{d^2\overline{D_R}}{d\alpha^2}>0$ and $\frac{d\overline{D_R}}{d\alpha}>0$. Thus, the first and second additive terms in the above equation are positive. The last one is also positive due to the assumption that $\overline{D_R}<\frac{1}{p^{det}_S}$. 
\end{proof}

Due to the convexity of $\tau$, we can use off-the-shelf numerical methods, e.g., golden section
search \cite{dl:chong}, to solve \eqref{eqminstatic}, as detailed in Algorithm \ref{algorithm1}.
\begin{algorithm} [t!]
\small{
  \caption{delay-optimal static policy}
  \begin{algorithmic}[1]   
  \IF  {\eqref{eqcond} holds}
   \STATE No-cooperation is delay-optimal
   \ELSE
   \STATE $a=0$ and $b=1$ 
    \WHILE{ ( $\frac{(b-a)}{b} \geq \epsilon$ ) } 
    \STATE $a_1 = a+0.382(b-a)$
    \STATE $b_1 = b-0.382(b-a)$
    \STATE Derive $\tau(a_1)$ and $\tau(b_1)$ from QBD process
    \STATE if $\tau(a_1)<\tau(b_1)$ then $b = b_1$
    \STATE if $\tau(a_1) \geq \tau(b_1)$ then $a = a_1$ 
    \ENDWHILE
    \ENDIF
  \end{algorithmic}
  \label{algorithm1} 
}
\end{algorithm}

\section{Optimal Dynamic Policies}
\label{sec:dynamic}
In this section, we derive the throughput-optimal and delay-optimal dynamic policies. In a dynamic policy, the mode selection decisions are made at the beginning of each slot based on the states of the data and energy buffers, i.e., $s$. Unlike the static policy, $\alpha_s$ is not the same for different states. Thus, the underlying QBD is not homogeneous in general. 

\subsection{Throughput-optimal dynamic policy}

Define $\{ \alpha_s \}_{s \in S}$ to be a dynamic policy, where $S$ is the set of all possible states at node R. Let $\overline{\alpha}=\sum_{s \in S} p_s \alpha_s$ denote the average DD probability corresponding to the dynamic policy, where $p_s$ is the probability of node R being in state $s$ at the beginning of a slot under policy $\alpha_s$. Note that Lemmas \ref{lem:stability}-\ref{lem:unstab} still hold in the dynamic case by replacing $\alpha$ with $\overline{\alpha}$, because all equations in these lemmas are based on average arrival and departure rates. 
The only part that needs revision, is the proof of ``if" part in Lemma \ref{lem:pb}. In this proof, we have assumed that when $\text{Q}_\text{d}$ is at the boundary of stability or unstable, it is always backlogged. Then, we have proved $p_b=0$. However, in the dynamic case, the data and energy arrivals can be controlled at each state. Thus, $\text{Q}_\text{d}$ is not necessarily backlogged at the boundary of stability. In other words, it can be empty in some slots (as will be shown in Proposition \ref{pro:thdy2}). Thus, in Lemma \ref{lem:complete}, we reprove the ``if" part of Lemma \ref{lem:pb} for the dynamic case. Consequently, based on Lemmas \ref{lem:stability}-\ref{lem:unstab}, we can conclude from Proposition \ref{pro:alpha} that the throughput-optimal dynamic policy also keeps $\text{Q}_\text{d}$ at the boundary of stability. Moreover, $\overline{\alpha}$ corresponding to the throughput-optimal dynamic policy is equal to $\alpha^{T}$ in \eqref{eq:alpha}. 

\begin{lemma}
If $\text{Q}_\text{d}$ is at the boundary of stability, then $p_b=0$.
\label{lem:complete}
\end{lemma}
\begin{proof}
We prove by contradiction. Suppose that there exists a dynamic policy which keeps $\text{Q}_\text{d}$ at the boundary of stability with $p_b>0$. Then, the average DD probability of such a policy, denoted by $\overline{\alpha}'$, is derived from \eqref{eq:alphat}. 
Note that $\overline{\alpha}$ in \eqref{eq:alphat} decreases with $p_b$. Thus, $\overline{\alpha}'$ is less than $\overline{\alpha}$ at $p_b=0$, i.e., $\alpha^{T}$. However, $\alpha^{T}$ is the throughput-optimal DD probability in the static policy, i.e., a stationary state-independent dynamic policy, that keeps the data buffer at the boundary of stability. Thus, we reach two different boundaries of stability, i.e., a contradiction.

\end{proof}

In the following proposition, we prove that the throughput-optimal dynamic policy is a simple threshold-based policy.


  
\begin{proposition}
The following threshold-based policy is a throughput-optimal policy
\begin{equation}
\alpha^{T}_{s}=\begin{cases}
1~;\quad   s\in\{(0,q_e)| ~q_e \geq e_{th}\} ,\\
0~;  \quad   s\in\{(0,q_e)| ~q_e < e_{th}\},\\
0~;  \quad  s\in\{(1,q_e)|~  0 \leq q_e \leq N \}.
\end{cases}
\label{eq:thop}
\end{equation}
\label{pro:thdy2}
for all choices of $e_{th}\leq N-b_{\text{max}}+1$. Consequently, $\overline{\alpha^T_s}=\sum_{s \in S} p_s \alpha^T_s=\alpha^T$ where $p_s$ is the probability of state $s$ under policy $\alpha^T_s$.
\end{proposition}
\begin{proof}
The policy described in \eqref{eq:thop} allows at most one data packet to be backlogged at node R at a time, thus stabilizing $\text{Q}_\text{d}$. If $p_b=0$ under $\alpha_{s}^{T}$ in addition, then $\alpha_{s}^{T}$ keeps $\text{Q}_{\text{d}}$ at the boundary of stability according to Lemma \ref{lem:pb}.  

To show $p_b$ is indeed $0$ under $\alpha_{s}^{T}$, we consider two cases under this policy. In the first case, $\text{Q}_\text{d}$ is empty. In this case, node R remains in the EH mode until the number of energy units exceeds $e_{th}$. Since $q_e<e_{th}$ and $e_{th} \leq N-b_{\text{max}}+1$, then $q_e \leq N-b_{\max}$. Thus, no blocking occurs. Once $q_e \geq e_{th}$, node R switches to the
DD mode until a packet is received. In the DD mode, no energy is harvested, and thus no energy units are blocked.
Therefore, the blocking probability, $p_b$, is zero in the first case. 

In the second case, there is one packet in $\text{Q}_\text{d}$. Thus, node R remains in EH mode until the packet is successfully transmitted. We show that the blocking probability in this case is also zero by observing the energy status before and after the first transmission of a packet.
A packet is transmitted by node R for the first time in the same slot of its arrival, if $q_e \geq K$ upon its arrival. Since node R is in the DD mode in this slot, no blocking occurs before the first transmission. But, if $q_e < K$ upon the arrival of the packet, node R harvests energy until $q_e \geq K$ before transmitting the packet for the first time. Therefore, $q_e < K$ before switching to the DD mode. Thus, due to our assumptions $N \geq 2K$ and $b_{\max}\leq K$, we have $q_e\leq N-b_{\max}$. Since at most $b_{\max}$ energy units are harvested in each slot, no energy unit is blocked. It can be concluded that the blocking probability is zero before the first transmission of a packet. After the first transmission, there remains at most $N-K$ energy units in the energy buffer regardless of the energy state at the moment of the packet arrival. Now, if the first transmission is unsuccessful, the packet is retransmitted and consumes $K$ energy units in each subsequent slot. Meanwhile, at most $b_{\max}$ energy units are harvested in each slot.
Due to the assumptions  $b_{\text{max}} \leq K$ and $N \geq 2K$, the number of energy units at the beginning of each slot does not exceed $N-K$. Thus, $p_b = 0$ after the first transmission of the packet. Therefore, $p_b=0$ in the second case, too. According to Lemma \ref{lem:pb}, $\alpha_{s}^{T}$ keeps $\text{Q}_\text{d}$ at the boundary of stability, and thus is throughput-optimal, i.e., $\overline{\alpha^T_s}=\alpha^T$. \end{proof} 

Here, we discuss the intuition why the throughput is maximized at the boundary of stability for both static and dynamic policies. At the boundary of stability, the data and energy arrival rates are balanced, i.e., all harvestable energy is consumed for successful transmission of all arrived packets at $\text{Q}_\text{d}$. Disturbing this balance leads to the blocking of the energy at $\text{Q}_\text{e}$ or instability of the data buffer, both lead to throughput degradation. This holds for both static and dynamic policies. The difference between the two policies is that the dynamic policy is able to control the buffer state more effectively, and thus can prevent $\text{Q}_\text{d}$ from being always backlogged at the boundary of stability.

\begin{remark}
\label{rem1}
It is worth noting that the throughput-optimal policy described in \eqref{eq:thop} is not unique. For example, when $b_{\max}=1$, any $e_{th}$ leads to the optimal throughput. In essence, adopting different $e_{th}$ in the range $[0, N-b_{\max}+1]$ leads to the same optimal throughput, albeit different idle time intervals at node R, i.e., the intervals node R waits till $q_e$ exceeds $e_{th}$. 

\end{remark}

\subsection{Delay-optimal dynamic policy}
\label{Delay-optimal dynamic policy}
According to \eqref{eq:delay}, the average transmission delay in a dynamic policy is $\tau(\alpha_s)=\frac{\overline{q_d(\alpha_s)}+1}{\lambda_S(\alpha_s)}$, where $\alpha_s$ denotes the applied policy at node R. 
Correspondingly, Lemma \ref{lem:lem5}  holds in the dynamic policy by replacing $\alpha$ with $\overline{\alpha}$. In Proposition \ref{pro:delay1}, we derive the delay-optimal dynamic policy, where $b_{\max}$ is set to $1$ for the time being.
\begin{proposition}
If $b_{\text{max}}=1$ and $\overline{D_R}|_{\overline{\alpha}\rightarrow 0}<\frac{1}{p^{det}_S}$, i.e., when cooperation is preferred to non-cooperation (see Lemma \ref{lem:lem5}), the delay-optimal policy $\alpha^{D}_s$ is given by
\begin{equation}
\alpha^{D}_s=\begin{cases}
1 \quad    ;\quad s=(0,N)\\
0       \quad    ;\quad \text{otherwise}. 
\label{eqpolicy1}
\end{cases}
\end{equation}
\label{pro:delay1}
\end{proposition}
\begin{proof}
Consider the following policy
\begin{equation}
\hat\alpha_s(\beta)=\begin{cases}
\beta \quad    ;\quad s=(0,N)\\
0  \quad    ;\quad \text{otherwise}. 
\label{eqpolicy2}
\end{cases}
\end{equation}

In Appendix \ref{app:two}, it is proved that $\overline{\alpha}$ corresponding to the above policy is an increasing continuous function of $\beta$. Therefore, since for extreme values $\beta=0$ and $\beta=1$, $\overline{\alpha}$ is respectively equal to $0$ and $\alpha^{T}$ according to Proposition~\ref{pro:thdy2}, any value of $\overline{\alpha} \in [0, \alpha^{T}]$ is achieved by a unique $\beta \in[0,1]$. 
 In the following, we denote this unique $\beta$  by $\beta=\beta(\overline \alpha)$.

Now, we prove that among all policies with the same average DD probability $\overline{\alpha}$, $\tau(\alpha_s)$ attains its minimum value under $\hat\alpha_s(\beta)$ where $\beta=
\beta(\overline \alpha)$. For the mentioned policies, $\lambda_s$ is fixed according to \eqref{eq:lem32}. Therefore, for minimizing  $\tau(\alpha_s)=\frac{\overline{q_d}(\alpha_s)+1}{\lambda_s(\alpha_s)}$, it is enough to find a policy which minimizes the nominator, i.e., $\overline{q_d}$. According to Little's law and \eqref{eqlamid}, we have
\begin{equation}
\overline{q_d}=\overline{D_R} \overline{\alpha} (1-p^{det}_{S}).
\label{eqqddr}
\end{equation}
Then, for a fixed $\overline{\alpha}$, the minimization of $\overline{q_d}$ is equivalent to the minimization of $\overline{D_R}$, which is the average queueing delay plus the average transmission time of the packets. 
The average queueing delay is zero under $\hat\alpha_s(\beta)$ since the packets are received only when $q_d=0$. Moreover, under $\hat\alpha_s(\beta)$, all packets observe full energy buffer upon their arrival. Thus, they all have the same average transmission time which we denote by $t_N$. As a result, $\overline{D_R}=t_N$. Now observe that $t_N$ is the minimum transmission time that can be experienced by a packet at node R because the energy buffer is full when node R starts to transmit the packet, and node R remains in the EH mode until the packet is successfully transmitted. Thus, node R has the maximum possible energy for initial transmission as well as probable retransmissions of the packet.  
 Putting all these together, we find that $\overline{D_R}$ obtains its minimum value under $\hat\alpha_s(\beta)$, and thus $\hat\alpha_s(\beta)$ with $\beta=\beta(\overline{\alpha})$ is the delay-optimal policy at a given $\overline{\alpha}$. Consequently, the minimum value of the nominator of $\tau$, i.e., $\overline{q_d}+1$, is given by
\begin{equation}
\overline{q_d}+1=t_N \overline{\alpha} (1-p^{det}_{S})+1.
\label{eqtn}
\end{equation}

So far, we have determined the structure of the delay-optimal policy at a given $\overline{\alpha}$, i.e., $\hat{\alpha}_s(\beta(\overline{\alpha}))$. Now, it remains to find the optimal $\overline \alpha$  which is equivalent to finding the optimal $\beta=\beta(\overline \alpha)$. 
In the following, we show that $\tau(\hat{\alpha}_{s}(\beta))$ is a decreasing function of $\beta=\beta(\overline \alpha)$. Then, we conclude that the optimal $\beta$ is $1$ and the optimal $\overline \alpha$ is $\alpha^{T}$. Thus, the proposition is proved. 

Using~\eqref{eq:lem31} and \eqref{eqtn} we have
\begin{equation}
\tau=\frac{\overline{\alpha}(1-p^{det}_{S})t_N+1}{\overline{\alpha}(1-p^{det}_{S})+p^{det}_{S}},
\label{eqtau1}
\end{equation}
and then 
\begin{equation}
 \frac{d\tau}{d\beta}=\frac{(1-p^{det}_{S})(p^{det}_S t_N-1)}{(\overline{\alpha}(1-p^{det}_{S})+p^{det}_{S})^2} ~\frac{d\overline{\alpha}}{d\beta}\,.
\label{eqdtau2}
\end{equation}
As mentioned before, it is proven in Appendix~\ref{app:two} that $d\overline \alpha/d\beta >0$. Then to verify that $d\tau/d\beta< 0$, we need to show that $t_N<\frac{1}{p_S^{det}}$. 
   
 According to~\eqref{eqlamid}, when $\alpha$ or $\overline{\alpha}$ tends to $0$, the average interarrival time of the packets at node R, which is $1/\lambda_{id}$, goes to infinity almost surely. In this case, the data buffer becomes empty and the energy buffer becomes fully backlogged before arrival of a new packet at $\text{Q}_\text{d}$. Thus, the packet is transmitted immediately upon its arrival at node R. Consequently, $\overline{D}_R|_{\overline{\alpha}\rightarrow 0}=t_N$. Since by our assumption $\overline{D_R}|_{\overline{\alpha} \rightarrow 0}<\frac{1}{p^{det}_S}$,
we have $t_N<\frac{1}{p^{det}_S}$. It is worth noting that regarding \eqref{eq:thop} and \eqref{eqpolicy1}, we deduce that for $b_{\max}=1$, the delay-optimal dynamic policy is also throughput-optimal.

\end{proof}

\begin{table*}[!t]
\begin{center}
\caption{Throughput-optimal and Delay-optimal Static and Dynamic Policies}
\scalebox{0.9}{
    \begin{tabular}{ | c | c | c |  }
    \hline
    Policy  & Throughput-optimal & Delay-optimal  \\ \hline
     Static & $ \alpha^{T} = \frac{E(\Gamma) p^{det}_{R}}{E(\Gamma) p^{det}_{R} +K(1-p^{det}_S)}$  & $\alpha^{D}$ is derived in Algorithm \ref{algorithm1} \\ \hline
     Dynamic &  

$\begin{aligned}
\alpha^{T}_{s}=\begin{cases}
1~;\quad   s\in\{(0,q_e)| ~q_e > e_{th}\} ,\\
0~;  \quad   s\in\{(0,q_e)| ~q_e \leq e_{th}\},\\
0~;  \quad  s\in\{(1,q_e)|~  0 \leq q_e \leq N \}.
\end{cases} \end{aligned} $ & $\begin{aligned} \alpha_s^{D}=\begin{cases}
1\quad ;\quad   s\in\{(0,q_e)| ~q_e > e_{th}\} ,\\
\beta \quad; \quad   s\in\{(0,q_e)| ~q_e= e_{th}\} ,\\
0 \quad  ; \quad \text{Otherwise},
\end{cases}
\end{aligned}
$\\
& $e_{th}\leq N-b_{\max}+1$ & $e_{th}$ and $\beta$ are derived in Algorithm \ref{alg:dyn}
   \\ \hline
    \end{tabular}}
    \label{tab:table2}

\end{center} 
\end{table*}

The above proposition describes the delay-optimal policy if $b_{\max}=1$, which is a threshold-based policy. 
The structure of the delay-optimal policy when $b_{\max}>1$ is more complicated. Nevertheless, by giving intuitions, we derive some properties of the delay-optimal policy. Similar to the case $b_{\max}=1$, for a given $\overline\alpha$, the throughput, $\lambda_S$, is constant. Thus, in this case, a delay-optimal policy should minimize $\overline{q_d}$ (see \eqref{eq:delay}). Consequently, according to \eqref{eqqddr}, it minimizes $\overline{D_R}$, i.e., the average queueing delay plus the average transmission time of the packet. The average transmission time of a packet will be less if no other packet is received during its transmission. Because, in this case all slots are dedicated to the EH mode, and thus, the maximum possible energy is harvested for transmission of the packet. Moreover, no packet is queued at node R. Thus, the queueing delay is zero. Therefore, it is deduced that for a given $\overline\alpha$, a delay-optimal policy should receive a packet after successfully transmitting the current packet, i.e., it keeps at most one packet at node R. Consequently, the delay-optimal policy, i.e., the policy related to the optimal $\overline{\alpha}$, has the same property. Now we focus on the policies with the structure as in \eqref{eq:thop} since these policies keep at most one packet at node R.

Our numerical results show that in a policy with a structure similar to \eqref{eq:thop} and for a typical EH profile (uniformly distributed between $0$ and $b_{\max}$), $\overline{\alpha}$ or equivalently the throughput, $\lambda_S$, increases by decreasing $e_{th}$ from $N$ to $N-b_{\max}+1$. This is intuitively due to the fact that by decreasing $e_{th}$, node R will be in the DD mode in higher energy levels, and thus, less energy units are blocked. Consequently, by decreasing $p_b$, $\overline \alpha$ increases according to \eqref{eq:alphat}. On the other hand, by increasing $\overline{\alpha}$, the arrival rate to $Q_d$ and as a result the average queue length, $\overline{q_d}$, increases. Since both the nominator ($\overline{q_d}$) and denominator ($\lambda_S$) in \eqref{eq:delay} increase, there is a chance that $\tau$ decreases for some $N-b_{\max}+1\leq e_{th}\leq N$. This is unlike the case $b_{\max}=1$ where the throughput is constant for every $e_{th}$. Moreover, if $e_{th}$ is chosen less than $N-b_{\max}+1$, $\overline{\alpha}$ and as a result, $\lambda_S$ remain constant according to Proposition \ref{pro:thdy2}. However, $\overline{D_R}$ and thus, $\overline{q_d}$ in \eqref{eqqddr} does not decrease since transmission of some of the packets begin at lower energy levels. Therefore, $\tau$ does not decrease. Based on this observations, we guess the structure of the delay-optimal policy for the case $b_{\max}>1$ in the following conjecture.  

\begin{conjecture} 
\label{conj:delaydy}
If the cooperation is helpful, i.e., $\overline{D_R}|_{\overline{\alpha}\rightarrow 0}<\frac{1}{p^{det}_S}$, the delay-optimal policy, $\alpha_s^{D}$, is given by
\begin{equation}
\alpha_s^{D}=\begin{cases}
1\quad ;\quad   s\in\{(0,q_e)| ~q_e > e_{th}\} ,\\
\beta \quad; \quad   s\in\{(0,q_e)| ~q_e= e_{th}\} ,\\
0 \quad  ; \quad \text{Otherwise},
\end{cases}
\label{eqstruc}
\end{equation}

for some $\beta \in [0,1]$ and $e_{th} \in \{N-b_{\max}+1,...,N\}$.
\end{conjecture}

It is worth mentioning that the DD probability in $e_{th}$ in \eqref{eqstruc} is considered as a   parameter, i.e., $\beta$, so that all the DD probabilities in $[0,\alpha^T]$ can be obtained continuously.

We now propose an algorithm to find the optimal values of $\beta$ and $e_{th}$. Consider the Markov chain related to $\alpha^{D}_s$ in \eqref{eqstruc}. The corresponding transition probabilities are written as in \eqref{eqtp1} and \eqref{eqtp2} by replacing $\alpha_s$ with $\alpha^{D}_s$.
Then, $\overline{\alpha}$ and $\overline{q_d}$ are derived as functions of $e_{th}$ and $\beta$:
\begin{equation}
\overline{\alpha}(e_{th},\beta)= \sum^{N}_{j=e_{th}+1}\pi_{(0,j)}+\beta \pi_{(0,e_{th})},
\label{eqalphabar}
\end{equation}
\begin{equation}
\begin{aligned}
&\overline{q_d}(e_{th},\beta) =\\
 &\frac{1}{2}(1-p^{det}_S)\left(\sum^{N}_{j=e_{th}+1}\pi_{(0,j)}+\beta \pi_{(0,e_{th})}\right)+\sum^{N}_{j=0} \pi_{(1,j)},
\label{eq:qd}
\end{aligned}
\end{equation}
where $\pi_{(i,j)}$ is the steady state probability of state $(i,j)$. In \eqref{eqalphabar}, $\overline{\alpha}(e_{th},\beta)$ is the expected value of $\alpha$ under the policy \eqref{eqstruc}.
In \eqref{eq:qd}, the first term accounts for the slots in which $Q_d$ is empty in the first subslot but switches to the DD mode and becomes backlogged in the second subslot. Likewise, the second term stands for the slots in which $\text{Q}_\text{d}$ is backlogged from the beginning of the slot. To find optimal $e_{th}$ and $\beta$, the following minimization is solved 
\begin{equation}
 \begin{aligned}
 & \underset{\substack{e_{th}\in \{N-b_{\text{max}}+1,.., N\},\\ \beta \in(0,1]  }} {\text{min}}
  \tau(e_{th},\beta)=\frac{\overline{q_{d}}(e_{th},\beta)+1}{(1-p^{det}_{S})\overline{\alpha}(e_{th},\beta)+p^{det}_S}~. \\
 \end{aligned}
 \label{eqminfinal}
\end{equation}

We prove in Appendix \ref{app:convex} that $\tau(e_{th},\beta)$ is a convex function of $\beta$. Then, we can use the golden section search method in Algorithm \ref{alg:dyn} to find the optimum value of $\beta$ for a given $e_{th}$. Meanwhile, the optimum value of $e_{th}$ is derived by exhausive search in $\{N-b_{\max}+1,...,N\}$. Table \ref{tab:table2} summarizes the throughput-optimal and delay-optimal static and dynamic policies. 

\begin{algorithm} [t!]
\small{
  \caption{Delay-optimal dynamic policy}
  \begin{algorithmic}[1]   
   \STATE $e_{th}=N$ and $\tau_{\min}=1/p^{det}_S$ (i.e., non-cooperation delay)
    \WHILE{ $e_{th} \geq N-b_{\text{max}}+1$} 
     \STATE $\beta=0$ and $\beta'=1$
    \WHILE{ ( $\frac{(\beta'-\beta)}{\beta'} \geq \epsilon$ ) } 
    \STATE $\beta_1 = \beta+0.382(\beta'-\beta)$
    \STATE $\beta'_1 = \beta'-0.382(\beta'-\beta)$
    \STATE Compute $\tau(e_{th},\beta_1)$ and $\tau(e_{th},\beta'_1)$ from QBD process
    \STATE if $\tau(e_{th},\beta_1)<\tau(e_{th},\beta'_1)$ then $\beta' = \beta'_1$
    \STATE if $\tau(e_{th},\beta_1) \geq \tau(e_{th},\beta'_1)$ then $\beta = \beta_1$ 
    \ENDWHILE
    \IF{$\tau(e_{th},\beta) \leq \tau_{\min}$ }
    \STATE $\tau_{\min}=\tau(e_{th},\beta)$
    \ENDIF
    \STATE $e_{th} = e_{th}-1$
     \ENDWHILE
  \end{algorithmic}
  \label{alg:dyn} 
}
\end{algorithm}

It is worth mentioning that if $N$ is large enough compared to $b_{\text{max}}$, the average transmission time of the packets that are received in energy states between $N-b_{max}+1$ and $N$ do not vary considerably (i.e., they are approximately equal to $t_N$). Therefore, \eqref{eqtau1} holds and thus, in order to minimize $\tau$, $e_{th}$ should be chosen such that $\overline{\alpha}$ or equivalently the throughput is maximized, i.e., $e_{th}=N-b_{max}+1$, according to Proposition \ref{pro:thdy2}. Thus, in this case, the throughput-optimal policy which is equivalent to $e_{th}=N-b_{max}+1$ and $\beta=1$ in \eqref{eqstruc}, is a good approximation for the delay-optimal policy. This has been shown by numerical results in the next section.

\section{Numerical Results}
In this section, we validate our analysis through numerical simulations. We also compare optimal dynamic and static policies to investigate the advantages of dynamic policies over the static ones in different conditions. The energy arrival process is modeled by a discrete uniform process in the interval $[0,b_{\text{max}}]$. The typical values of the parameters are listed in Table \ref{tab:table3}.
\begin{table}[!t]
\begin{center}
\caption{Typical Values of Parameters}
\scalebox{0.85}{
    \begin{tabular}{ | c | c | c | c | }
    \hline
     parameter & Typical value & parameter & Typical value \\ \hline
     $p^{det}_{S}$ &   $0.3$ & $p^{det}_{R}$ &   $0.9$ \\ \hline
     $b_{\text{max}}$  &   $5$ &  $K$          &   $10$    \\ \hline
     $N$           &   $100$  &  $\epsilon$ in Algorithm \ref{algorithm1} & $0.01$ \\ \hline           
    \end{tabular}}
    \label{tab:table3}

\end{center} 
\end{table} 

\begin{figure}[!t]
\begin{center}
\includegraphics[height=120pt,width=\columnwidth]{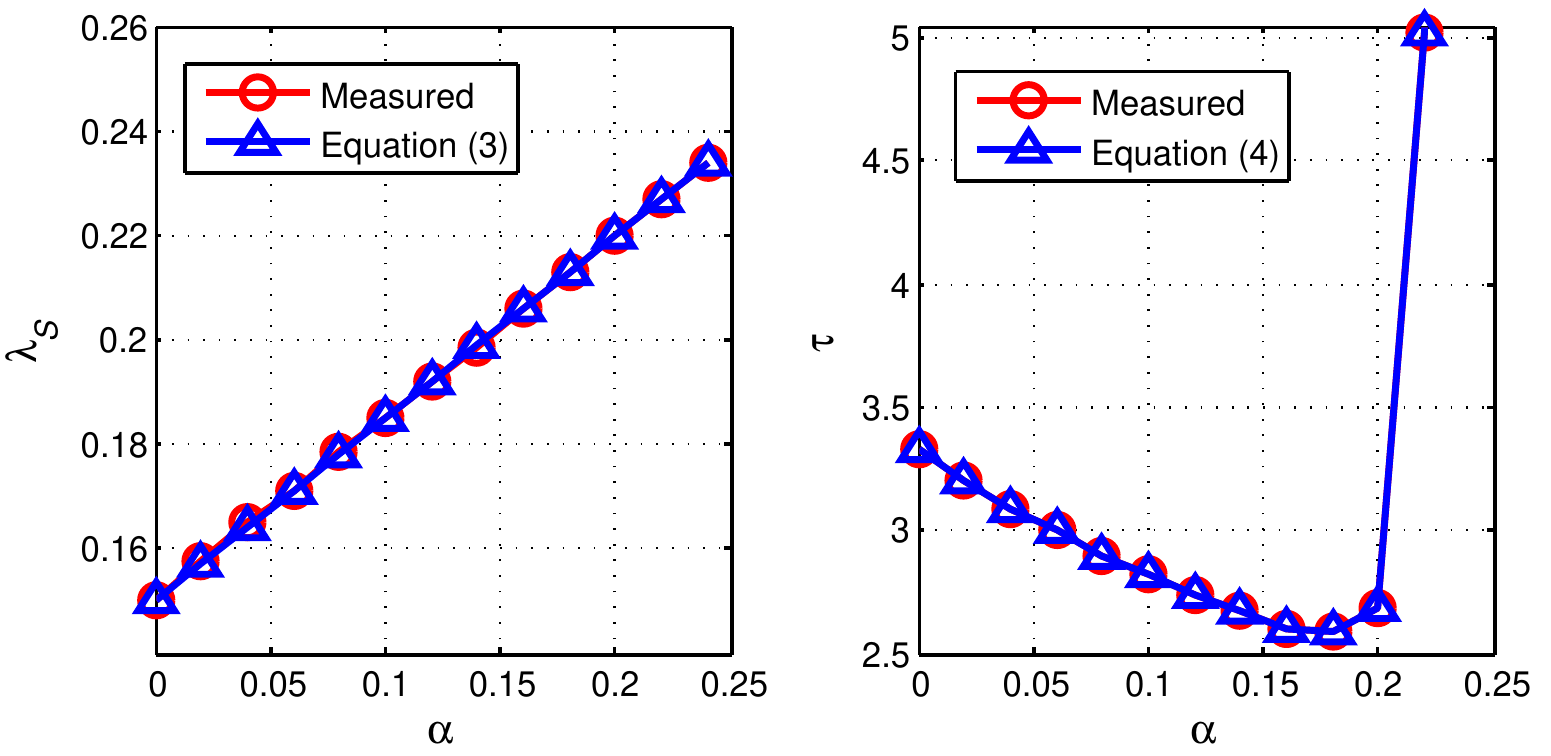}
\caption{Average transmission delay and throughput of the source versus DD probability ($\alpha$)}
\label{fig:confirm}
\end{center}
\end{figure}
In Fig. \ref{fig:confirm}(a) and (b), we validate \eqref{eq:thr} and \eqref{eq:delay} by comparing them with the average delay and throughput obtained by simulations. The figures show that the analysis matches with the simulations very well. Moreover, in Fig. \ref{fig:confirm2}, we verify our assumption in Section \ref{delay-optimal static policy} that the average delay at node R, i.e., $\overline {d_R}$, is a convex and increasing function of $\alpha$.
\begin{figure}[!t]
\begin{center}
\includegraphics[height=150pt,width=0.9\columnwidth]{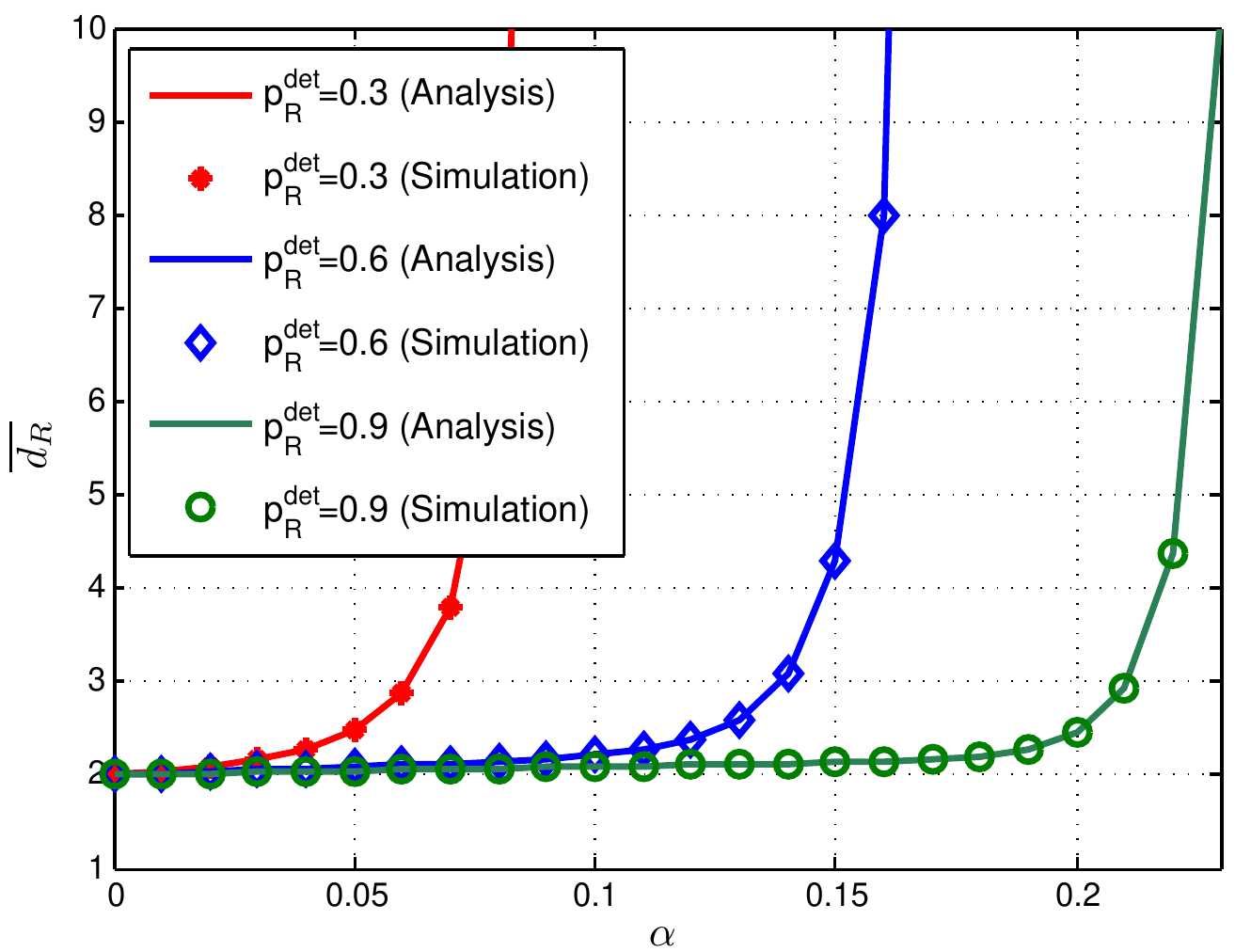}
\caption{Average transmission delay at node R versus DD probability ($\alpha$).}
\label{fig:confirm2}
\end{center}
\end{figure}

\begin{figure}
\centering
{\includegraphics[height=340pt,width=\columnwidth]{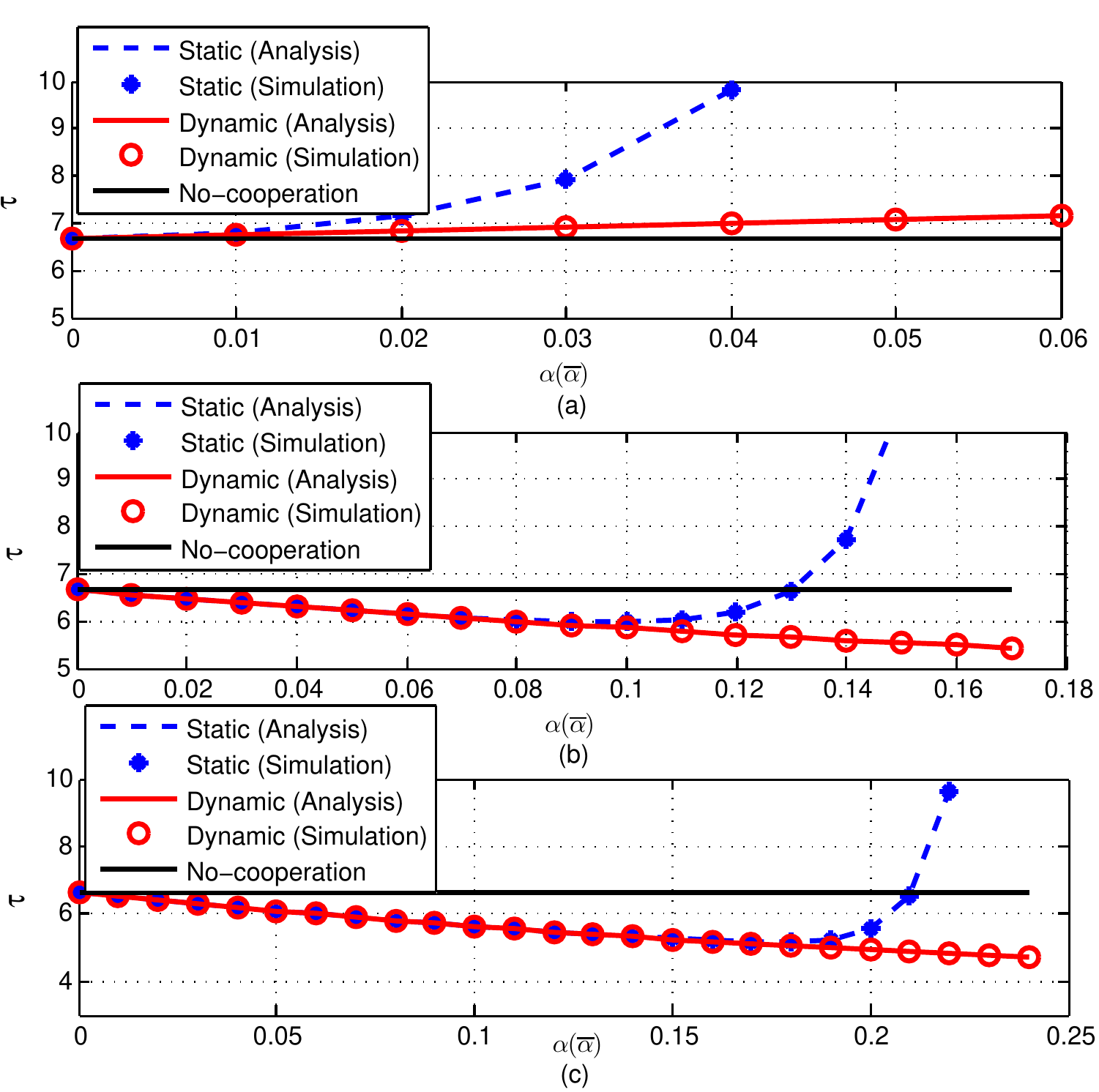}}
\\
\caption{Average transmission delay of source packets vs. DD probability at the relay node (a) $p^{det}_R=0.2$ (b) $p^{det}_R=0.6$ (c) $p^{det}_R=0.9$}
\label{fig:alpha}
\end{figure}

In Fig. \ref{fig:alpha}, the average transmission delay of source packets, $\tau$, is plotted versus $\alpha$ of the static policy, and $\overline{\alpha}$ of the dynamic policy, respectively. Note that $\alpha$ can be directly translated to throughput according to \eqref{eq:lem32}. The figures show that $\tau$ is a convex function of $\alpha$ or equivalently the data arrival rate at node R (see \eqref{eqlamid}).
Moreover, as shown in Fig. \ref{fig:alpha}(a), when $p^{det}_R$ is low, the number of retransmissions of a packet at node R is so large that the cooperation does not improve the average transmission delay. However, if $p^{det}_R$ is increased to $0.6$ and $0.9$ as in Fig. \ref{fig:alpha}(b) and (c), cooperation reduces the average transmission delay for a certain range of DD probabilities. In these cases, when $\alpha$ is not very large ($\alpha<0.1$ in Fig. \ref{fig:alpha}(b) and $\alpha<0.19$ in Fig. \ref{fig:alpha}(c)), the performance of static and dynamic policies are the same. This is due to the fact that no queue is yet formed at node R in the static policy. Thus, its performance is similar to the optimal dynamic policy, which accepts at most one packet at node R according to conjecture \ref{conj:delaydy}. In addition, in this case, the delay decreases by increasing $\alpha$, because more packets are transmitted through the better physical channel ($p^{det}_R$). However, when $\alpha$ becomes too large ($\alpha>0.1$ in Fig. \ref{fig:alpha}(b) and $\alpha>0.19$ in Fig. \ref{fig:alpha}(c)), the queueing delay at node R of the static policy becomes noticeable, so that an explosive growth in the average transmission delay is observed. This is unlike the optimal dynamic policy that serves at most one packet at a time and thus, prevents a queue build-up at node R. 
 It is worth noting that in Fig. \ref{fig:alpha}(b) and (c), the throughput-optimal dynamic policy and the delay-optimal dynamic policy are the same. However, when static policies are deployed, the throughput and delay optimality cannot be attained at the same time.  

\begin{figure}
\centering
{\includegraphics[height=180pt,width=0.9\columnwidth]{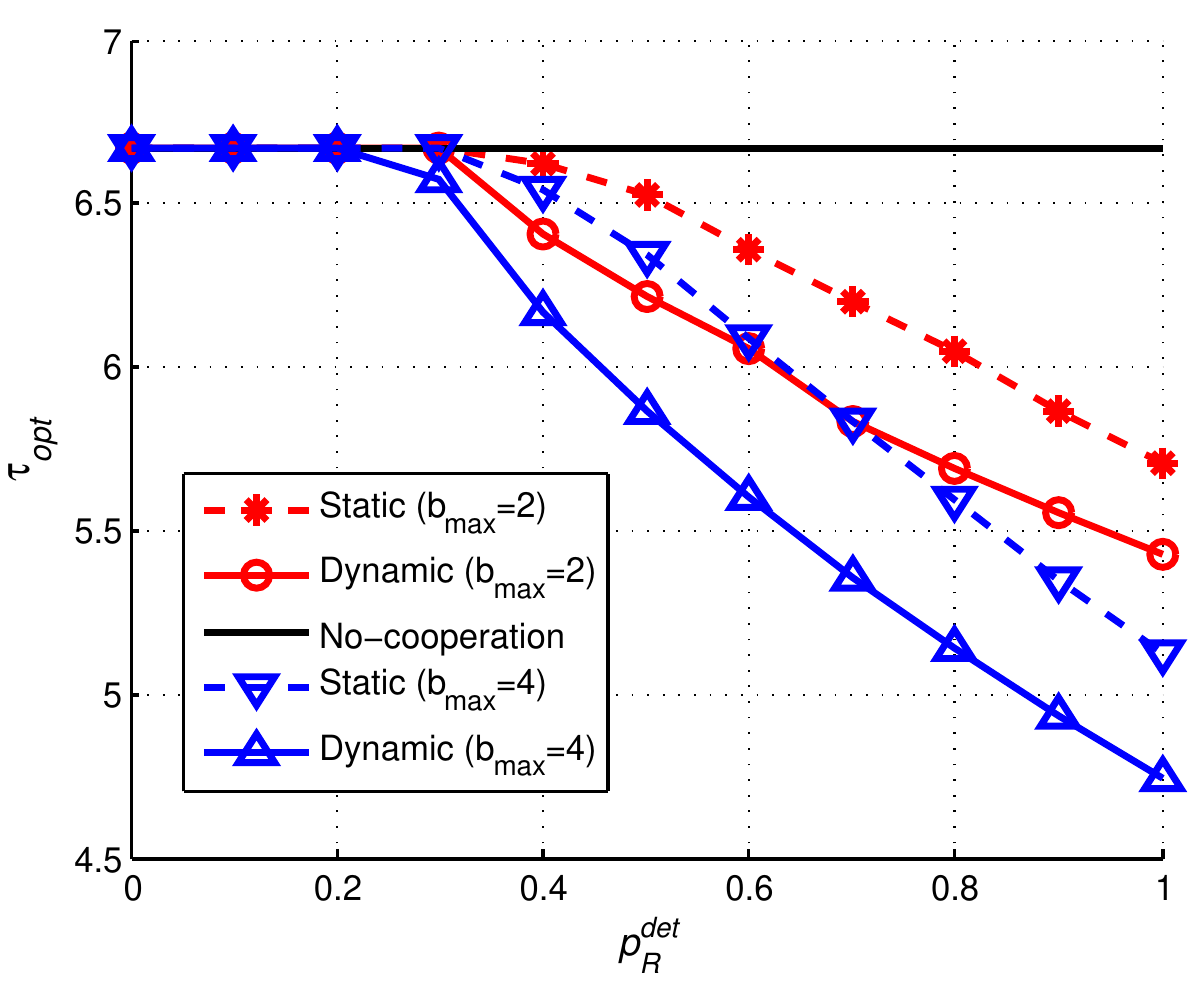}}
\\
\caption{Optimum average transmission delay of source packets vs. detection probability of R-D channel}
\label{fig:det}
\end{figure}

In Fig. \ref{fig:det}, the optimal average transmission delay of source packets is plotted versus $p^{det}_R$. When $p^{det}_R$ is small, both static and dynamic policies decide not to cooperate due to large transmission delay at node R. However, when $p^{det}_R$ is large, cooperation leads to much better performance than no-cooperation. Moreover, the increase in the energy arrival rate at node R leads to a decrease in the average transmission delay as expected.

\begin{figure}
\centering
{\includegraphics[height=180pt,width=0.9\columnwidth]{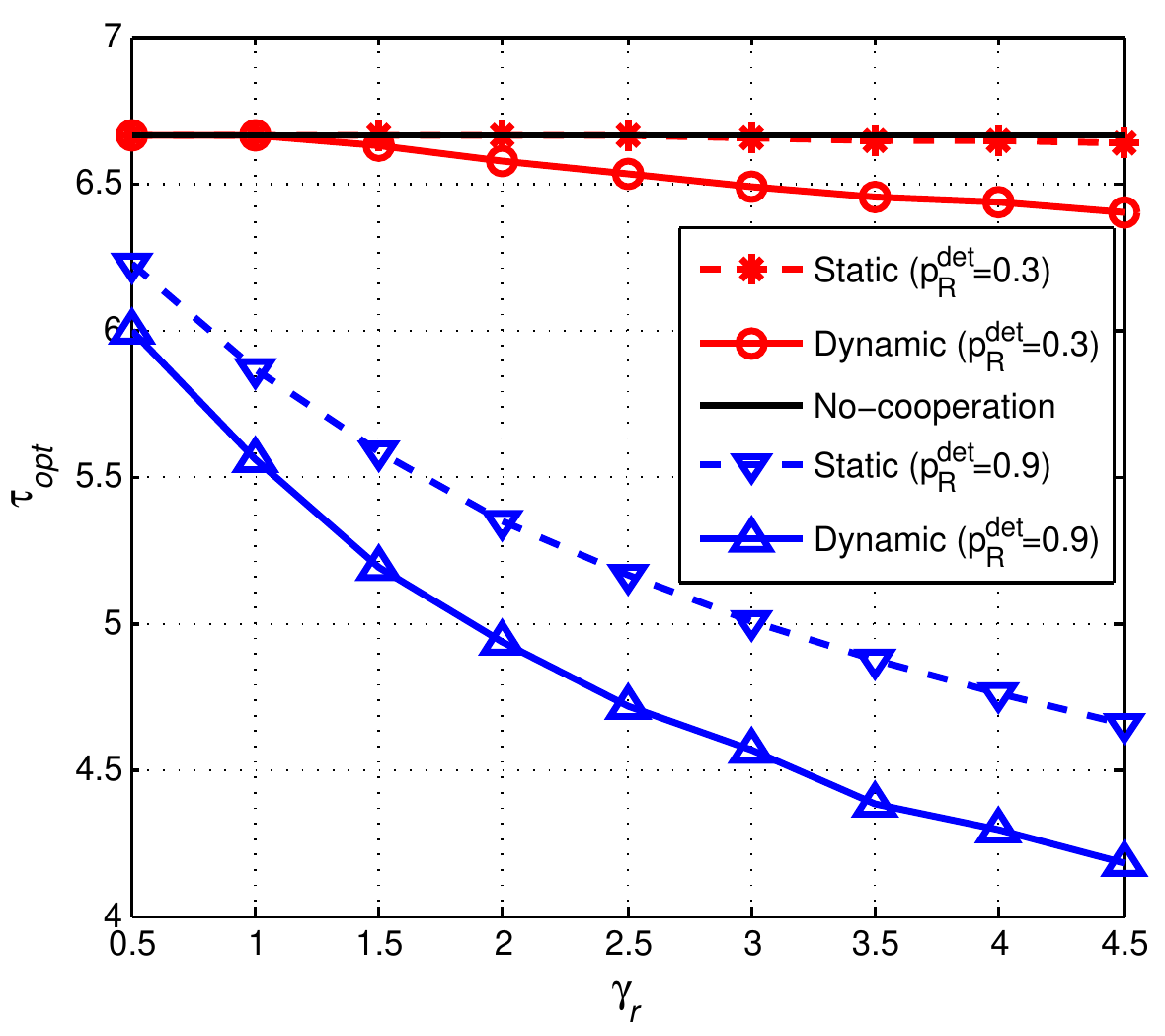}}
\\
\caption{Optimum average transmission delay of source packets vs. energy arrival rate at node R}
\label{fig:energy}
\end{figure} 

In Fig. \ref{fig:energy}, the optimal average transmission delay is plotted versus the energy arrival rate at node R. As the energy arrival rate increases, the average transmission time of a packet decreases. Thus, the improvement over no-cooperation scenario is more noticeable. Moreover, when $p^{det}_R$ is large, the improvement of dynamic policy over the static one is more significant. This is due to the lower average transmission time at node R, which leads to a more efficient dynamic policy.
 
\begin{figure}
\centering
{\includegraphics[height=180pt,width=0.9\columnwidth]{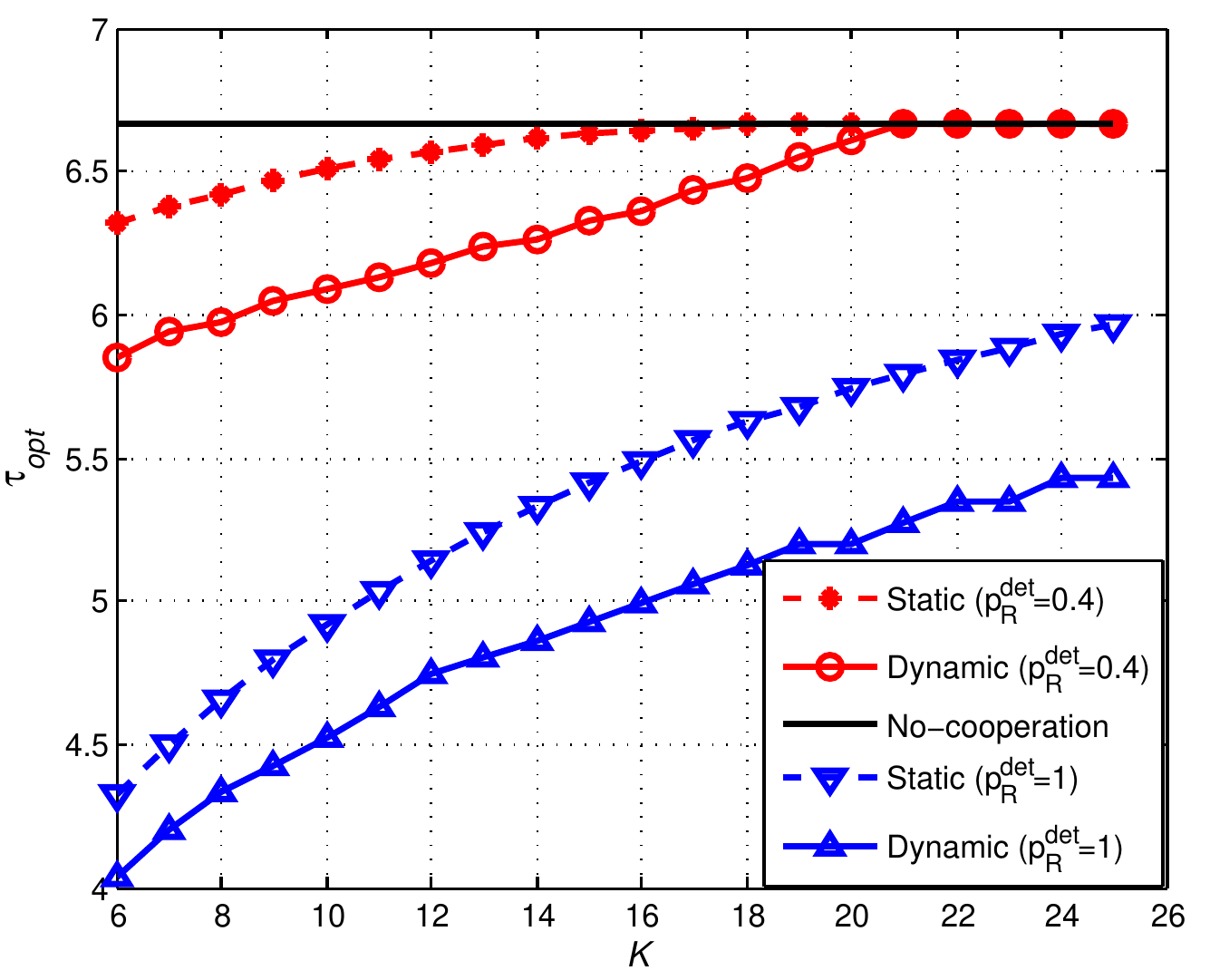}}
\\
\caption{Optimum average transmission delay of source packets vs. the number of energy units used at node R for transmission of a packet}
\label{fig:k}
\end{figure} 

In Fig. \ref{fig:k}, the optimal average transmission delay is plotted versus the number of energy units that are needed for transmission of a single packet, i.e., $K$. As shown in Fig. \ref{fig:k}, the average transmission delay increases with $K$. This is because a packet must wait for longer time until enough energy is accumulated at the energy buffer. Moreover, when $p^{det}_R$ is small and $K$ is large, non-cooperation is optimal for both the static and dynamic policies. This is because of both the high retransmission rate at node R (due to small $p_R^{det}$) and long waiting time for each retransmission (due to large $K$).

\begin{figure}
\raggedleft
{\includegraphics[height=160pt,width=\columnwidth]{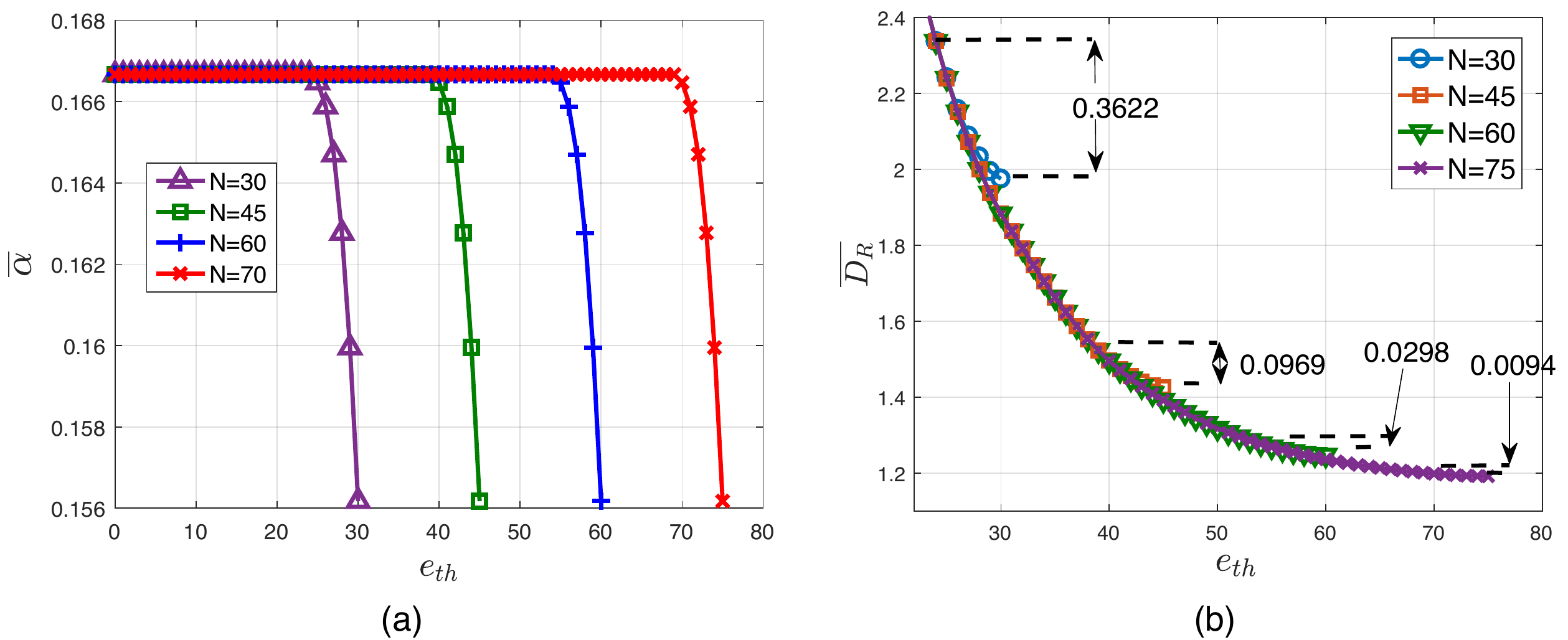}}
\\
\caption{(a) Average DD probability, $\overline{\alpha}$, vs. $e_{th}$ in \eqref{eq:thop},  (b) Average delay at node R, $\overline{D_R}$, vs. $e_{th}$ in \eqref{eqstruc} ($p^{det}_R=0.6,p^{det}_S=0.3,K=15,b_{max}=7$)}
\label{fig:combine}
\end{figure} 


In Fig. \ref{fig:combine}(a), the average DD probability, $\overline{\alpha}$, corresponding to $\alpha^T_s$ in \eqref{eq:thop} is plotted versus the energy threshold  to switch to the DD mode in $\alpha^T_s$, i.e., $e_{th}$. As it can be seen, for all values of $N$, $\overline{\alpha}$ and consequently the throughput increase by decreasing $e_{th}$ from $N$ to $N-b_{\max}+1$. Then it remains constant for $e_{th}\leq N-b_{\max}+1$ as proved in Proposition \ref{pro:thdy2}. Note that since the optimal throughput, i.e., $\alpha^T$ in \eqref{eq:alpha}, is independent of $N$, the maximum throughput for all values of $N$ is the same. 
In Fig. \ref{fig:combine}(b), the average delay at node R, i.e., $\overline{D_R}$, corresponding to $\alpha^D_s$ in \eqref{eqstruc} is plotted versus $e_{th}$ when $\beta=1$. It can be observed that $\overline{D_R}$ decreases as $e_{th}$ increases. This is due to the fact that for larger $e_{th}$, the packets are received in higher energy levels which leads to less average transmission time of the packets. Moreover, the difference between $\overline{D_R}$ at $e_{th}=N$ and $e_{th}=N-b_{\max}+1$ is shown in Fig. \ref{fig:combine}(b) for different values of $N$. As it can be observed, this difference becomes less as $N$ increases. In sufficiently large $N$, the difference is negligible which means that $\overline{D_R}$ remains almost constant for $e_{th} \in \{N-b_{\max}+1,...,N\}$. In this case, it can be inferred that the average transmission time of the packets does not change considerably when they start their transmission at $q_e\in \{N-b_{\max}+1,..., N \}$, i.e., they are approximately equal to $t_N$. Thus, \eqref{eqtau1} holds and similar to the proof of Proposition \ref{pro:delay1}, $\tau$ is minimized when $\overline{\alpha}$ and consequently the throughput have their maximum values. This happens when $e_{th}=N-b_{\max}+1$ and $\beta=1$ in \eqref{eqstruc}, according to Proposition \ref{pro:thdy2}. Therefore, in this case, the delay-optimal policy is also throughput-optimal.

\begin{table}[!t]
\begin{center}
\caption{Optimal Energy Threshold in $\alpha^D_s$ in \eqref{eqstruc} ($p^{det}_S=0.3$, $N=45$, $K=15$, $b_{\max}=7$),\protect\linebreak
(Energy Threshold in Throughput-optimal Policy, $\alpha^T_s$ in \eqref{eq:thop}, is $e_{th}\leq 39$)}
\scalebox{0.9}{
    \begin{tabular}{ | c | c | c || c | c | c |}
    \hline
     $p^{det}_R$ & $e_{th}$  & $\tau$ & $p^{det}_R$  & $e_{th}$ & $\tau$ \\ \hline
     $0.45$ &   $45$ & $3.1912$ &   $0.7$ & $41$ & $2.6333$ \\ \hline
     $0.5$ &   $44$ & $3.0459$ &   $0.8$ & $40$ & $2.4942$ \\ \hline
     $0.6$ &   $42$ & $2.8119$ &   $0.9$ & $39$ & $2.3825$ \\ \hline           
    \end{tabular}}
    \label{tab:table4}

\end{center} 
\end{table} 

In Table \ref{tab:table4}, we list $e_{th}$ in \eqref{eqstruc} for the delay-optimal policy in different detection probabilities at node R. It can be seen that $e_{th}$ decreases from $N=45$ to $N-b_{\max}+1=39$ when $p_R^{det}$ increases.
Also note that for any $p^{det}_R$, any $e_{th}\leq N-b_{\max}+1=39$ is throughput-optimal in \eqref{eq:thop}, according to Proposition \ref{pro:thdy2}.

\section{Conclusion}
We have derived the throughput-optimal and delay-optimal static and dynamic TS policies in a cooperative network with a WEH relay. We have proved that the throughput-optimal policy is obtained when the data buffer at the relay is at the boundary of stability, where the blocking probability at energy buffer is zero and all the harvested energy is consumed for transmission of data packets. Unlike the static policy, we have proved that the throughput-optimal dynamic policy has a threshold-based structure that keeps at most one data packet at the relay node. The delay-optimal static policy has been derived by analyzing a underlying QBD process. Likewise, we have shown that the delay optimal dynamic policy is threshold-based. In particular, the policy allows at most one data packet to be backlogged at node R. Our analysis has been validated by extensive simulations, through which have also investigated the performance of the proposed policies under various settings of system parameters.
\appendices
\section{Transition Probabilities of the Underlying QBD Process}
\label{ap:qbd}

Define $\mathbf{M}$ to be the transition matrix of the energy buffer state in the second subslot given that the data buffer is backlogged at the end of the first subslot. Then, $\mathbf{M}_{ij}$ is written as
\begin{equation}
\mathbf{M}_{ij}=\begin{cases}
1  \quad ; \quad\quad ~ ( i<K , j=i )~ \text{or}~ ( i\geq K, j=i-K) , \\
0  \quad  ;\quad\quad\quad\quad\quad\quad\quad\quad \text{otherwise.} \\
\end{cases}
\label{eqM2}
\end{equation}
The above equation indicates that the transmission occurs in the second subslot only when at least $K$ energy units is available in the energy buffer. Otherwise, the number of energy units does not change. Likewise, define $\mathbf{T}$ to be the transition matrix of the energy buffer state in the first subslot when node R is in the EH mode. $\mathbf{T}_{ij}$ is derived as
\begin{equation}
\mathbf{T}_{ij}=\begin{cases}
0 \quad \quad \quad\quad ~~ ; \quad\quad ~  j<i ,\\
\gamma_{j-i} \quad\quad\quad ~~ ; \quad\quad  i \leq j < N ,  \\
\sum^{b_{\text{max}}}_{l=N-i}\gamma_{l} \quad ; \quad\quad  j = N , \\
\end{cases}
\label{eqT1}
\end{equation}
According to the above equation, the number of energy units does not decrease in the first subslot, since node R does not transmit. Also, due to limited capacity of $\text{Q}_\text{e}$, i.e., $N$, the energy buffer becomes full if more that $N-i$ energy units arrive when the energy state is $i$.

Let $\alpha_s$ denote the probability of switching to the DD mode at state $s$, where $s\in\{0,1,...\}\times\{0,1,...,N\}$. Then, the transition probability from state $s=(0,i)$ at the onset of a slot to state $(l,j)$ at the onset of the next slot, represented by $P_{0i\rightarrow lj}$, is written as
\begin{equation}
\begin{aligned}
P_{0i\rightarrow lj}=\begin{cases}
\alpha_s(1-p^{det}_{S})p^{det}_{R} \mathbf{M}_{ij}+\alpha_s p^{det}_{S} \mathbf{I}_{ij}+\\
(1-\alpha_s)\mathbf{T}_{ij} \quad  &;\quad  l=0~, \\
\alpha_s(1-p^{det}_{S})(1-p^{det}_{R}) \mathbf{M}_{ij}~~&; \quad   l=1~, \\
\end{cases}
\label{eqtp1app}
\end{aligned}
\end{equation}
where $\mathbf{I}_{ij}$ is the $ij$-th component of the identity matrix. Note that when node R switches to the DD mode, the energy state does not change in the first subslot. In the above equation, the first case states that node R remains empty in three conditions. First, it switches to the DD mode in the first subslot, receives a packet and transmits it successfully in the second subslot. In this condition, node R becomes backlogged at the end of the first subslot, and thus its energy state changes according to $\mathbf{M}$. Second, it switches to the DD mode but does not receive any packet. Therefore the energy state does not change. Third, it harvests energy in the first subslot. Moreover, the second case in \eqref{eqtp1app} indicates that the number of data packets increases by one if node R switches to the DD mode and receives a packet, but  does not transmit it successfully in the second subslot.    

Let $\mathbf{B}$ denote the transition matrix of the energy buffer state given that node R is backlogged at the onset of a slot and selects EH mode. Then, $\mathbf{B}=\mathbf{T} \times \mathbf{M}$. The transition probability from state $s=(l,i)$ to state $(l',j)$ ($l>0$), represented by $P_{li\rightarrow l'j}$  is written as 
\begin{equation}
\begin{aligned}
P_{li\rightarrow l'j}=\begin{cases}
\alpha_s((1-p^{det}_{S})p^{det}_{R}+p^{det}_{S}(1-p^{det}_{R}))\mathbf{M}_{ij}+\\(1-\alpha_s)(1-p^{det}_{R})\mathbf{B}_{ij} \qquad \qquad \quad ~;~ l'=l~, \\
\alpha_s(1-p^{det}_{S})(1-p^{det}_{R})\mathbf{M}_{ij} ~ \qquad \quad ; ~  l'=l+1, \\
\alpha_s p^{det}_{S}p^{det}_{R}\mathbf{M}_{ij}+(1-\alpha_s)p^{det}_{R}\mathbf{B}_{ij} ~ ; ~ l'=l-1. 
\end{cases}
\end{aligned}
\label{eqtp2app}
\end{equation}
In the above equation, the first case indicates that when $l>0$, the number of data packets does not change in three situations. First, the DD mode is selected, a packet is received in the first subslot, and a packet is transmitted successfully in the second subslot. Second, the DD mode is selected, but neither a packet is received in the first subslot nor a packet is transmitted successfully in the second subslot. Third, node R switches to EH mode and does not transmit a packet successfully in the second subslot. In this situation, the energy buffer state changes according to $\mathbf{B}$. Other cases are written in a similar way. 
 
\section{Proof of $\frac{d\overline \alpha}{d\beta}>0$ in Proposition \ref{pro:delay1}}
\label{app:two}
Consider the Markov chain corresponding to the policy $\hat\alpha_s(\beta)$ in \eqref{eqpolicy2}, in which $s=(l,i)\in \{0,1\} \times \{0,1,...,N\}$. The transition probabilities of the Markov chain are derived as in \eqref{eqtp1} and \eqref{eqtp2} by replacing $\alpha_s$ with $\hat\alpha_s(\beta)$. Let $\mathbf{F}$ denote the corresponding transition matrix. Note that  
the transition probabilities in $\mathbf F$ are all independent of $\beta$ except those in ${N+1}^{th}$ row, which correspond to transitions from state $s=(0,N)$ to other states. These transition probabilities  are linear functions of $\beta$. Thus $\mathbf{F}$ can be written as 
\begin{equation}
 \begin{aligned}
 \mathbf{F}=\mathbf{\hat{F}}+\beta \mathbf{\tilde{F}}
 \end{aligned}
 \label{eqf}
\end{equation}
where all entries in $\mathbf{\tilde{F}}$ are zero except those in ${N+1}^{th}$ row. Also, let $\boldsymbol{\nu}$ denote the eigenvector of $\mathbf{F}$ associated to eigenvalue $1$ that is normalized in such a way that the ${N+1}^{th}$ coordinate of $\boldsymbol{\nu}$ is equal to $1$, i.e., $\boldsymbol{\nu}_{N+1}=1$. Note that with this normalization, $\boldsymbol{\nu}$ is a scaled version of the probability vector. Then, by substituting $\mathbf{F}$ from \eqref{eqf} to $\boldsymbol{\nu}\mathbf{F}=\boldsymbol{\nu}$ and taking the first derivative, we have
\begin{equation}
\begin{aligned}
\frac{d\boldsymbol{\nu}}{d\beta}=\boldsymbol{\nu}\mathbf{\tilde{F}}+\frac{d\boldsymbol{\nu}}{d\beta}\mathbf{\hat{F}}+\beta \frac{d\boldsymbol{\nu}}{d\beta} \mathbf{\tilde{F}}.
\end{aligned}
\label{eqfirstder}
\end{equation}
Then by the normalization $\boldsymbol{\nu}_{N+1}=1$, we have $\frac{d\boldsymbol{\nu}_{N+1}}{d\beta}=0$. As a result, since all rows of $\tilde{\mathbf{F}}$ except row $N+1$ are zero, the third term in \eqref{eqfirstder} vanishes. 

Next taking another derivative we obtain
\begin{equation}
\begin{aligned}
\frac{d^2\boldsymbol{\nu}}{d\beta^2}=\frac{d\boldsymbol{\nu}}{d\beta}\mathbf{\tilde{F}}+\frac{d^2\boldsymbol{\nu}}{d\beta^2}\mathbf{\hat{F}} =  \frac{d^2\boldsymbol{\nu}}{d\beta^2}\mathbf{\hat{F}}.
\end{aligned}
\label{eqsecder}
\end{equation}
This means that $\frac{d^2\boldsymbol{\nu}}{d\beta^2}$ is either zero or an eigenvector of $\hat{\mathbf{F}}$ with eigenvalue $1$. In the latter case, $\frac{d^2\boldsymbol{\nu}}{d\beta^2}=c\boldsymbol{\hat{\pi}}$ where $\boldsymbol{\hat{\pi}}$ is the steady-state probability vector of $\mathbf{\hat{F}}$ and $c$ is a constant. Since $\frac{d^2\boldsymbol{\nu_{N+1}}}{d\beta^2}=0$, then $c=0$. As a result, $\frac{d^2\boldsymbol{\nu}}{d\beta^2}=0$ and thus, $\boldsymbol{\nu}=\beta\mathbf{a}+\mathbf{b}$ where $\mathbf{a}_{N+1}$ is zero since $\boldsymbol{\nu}_{N+1}$ is assumed to be fixed and independent of $\beta$, i.e., $\nu_{N+1}=1$. Also, let $\boldsymbol{\pi}$ to be the steady-state probability vector of $\mathbf{F}$, then, we have $\boldsymbol{\pi}=\frac{\boldsymbol{\nu}}{\boldsymbol{\nu^T}\mathbf{1}}$. Thus, $\boldsymbol{\pi}_{N+1}$ is given by
\begin{equation}
\begin{aligned}
\boldsymbol{\pi}_{N+1}=\frac{\mathbf{b}_{N+1}}{\beta\,\mathbf{a}^{T}\mathbf{1}+\mathbf{b}^{T}\mathbf{1}}.
\end{aligned}
\label{eqstate}
\end{equation} 
Now, the average DD probability corresponding to policy $\hat\alpha_s(\beta)$ in \eqref{eqpolicy2} is $\overline{\alpha}=\beta \boldsymbol{\pi}_{N+1}= \frac{\beta\,\mathbf{b}_{N+1}}{\beta\,\mathbf{a}^{T}\mathbf{1}+\mathbf{b}^{T}\mathbf{1}}$. Then $\overline{\alpha}$ is an increasing function of $\beta$ since $\mathbf{b}^T\mathbf{1}>0$ ($\mathbf{b}$ is a scaled version of $\hat{\pi}$ since it corresponds to $\beta=0$). 
\section{Proof of Covexity of $\tau(e_{th},\beta)$ versus $\beta$}
\label{app:convex}
If Conjecture \ref{conj:delaydy} holds, then $\tau(e_{th},\beta)$, i.e., the average delay corresponding to $\alpha^{D}_s$ in  \eqref{eqstruc}, is a convex function of $\beta$.
In order to prove the convexity, we should prove the following inequality
\begin{equation}
\begin{aligned}
\theta \tau(e_{th},\beta_1)+(1-\theta)\tau(e_{th},\beta_2)>\tau(e_{th},\theta \beta_1+(1-\theta)\beta_2)~ ; \quad   \\ \forall \theta \in [0,1],
\label{eqproofconv}
\end{aligned}
\end{equation}
where $\beta_1, \beta_2 \in [0,1]$.
The left-hand side of \eqref{eqproofconv} is the average transmission delay in the policy obtained by time sharing between two policies with the same structure as $\alpha_s^{D}$ in \eqref{eqstruc} but different DD probabilities at state $(0,e_{th})$, i.e., $\beta_1$ and $\beta_2$ , respectively. Also, observe that the average DD probability of the time-sharing policy at state $(0,e_{th})$ is equal to $\theta \beta_1+(1-\theta)\beta_2$. However, it is not necessarily stationary. Because of the unichain Markovian structure of our problem, a stationary policy exists with the same delay as the delay of time-sharing policy \cite{dl:puterman}. On the other hand, the right-hand side of \eqref{eqproofconv} is the minimum delay among all stationary policies with the average DD probability equal to $\theta \beta_1+(1-\theta)\beta_2$. Thus, \eqref{eqproofconv} holds. This completes the proof.

\setstretch{1}

\end{document}